\documentclass[A4,10pt]{article}
\usepackage[margin=1in,dvips]{geometry}
\usepackage{graphicx,cancel}
\usepackage{graphics}
\usepackage{color, soul}
\usepackage{amsmath, amsthm, slashed}
\usepackage{amssymb}
\usepackage{marvosym}
\usepackage{rotating}
\usepackage{mathrsfs}
\usepackage{upgreek}
\usepackage{epsfig}
\usepackage{verbatim}
\usepackage[affil-it]{authblk}
\usepackage{rotating}
\usepackage{graphicx}
\usepackage{accents}
\usepackage{harmony, slashed}
\usepackage{tikz}

\usepackage{enumerate}

\newtheorem{definition}{Definition}[section]
\newtheorem{theorem}{Theorem}[section]
 
\newtheorem*{conjecture*}{Conjecture}

\newtheorem*{theorem*}{Theorem}
\newtheorem*{corollary*}{Corollary}
\newtheorem{proposition}{Proposition}[subsection]
\newtheorem{lemma}{Lemma}[subsection]
\newtheorem{remark}{Remark}[section]

\DeclareMathAlphabet\mathbfcal{OMS}{cmsy}{b}{n}

\title{Rough initial data and the strength of the blue-shift instability\\ on cosmological black holes with
$\Lambda>0$}

\author[1,2]{Mihalis Dafermos}
\author[2]{Yakov Shlapentokh-Rothman}

\affil[1]{\small University of Cambridge, Department of Pure Mathematics and Mathematical
Statistics, Wilberforce~Road,~Cambridge~CB3~0WA,~United~Kingdom\vskip.2pc \ }

\affil[2]{\small Princeton University, Department of Mathematics, Fine~Hall,~Washington~Road,~Princeton,~NJ~08544,~United~States~of~America\vskip.2pc \ }

\begin{document}

\maketitle

\begin{abstract}
We consider the wave equation on  Reissner--Nordstr\"om--de Sitter
and more generally Kerr--Newman--de Sitter black hole spacetimes with $\Lambda>0$. 
The strength of the blue-shift instability associated to the Cauchy horizon of
these spacetimes has been the subject of much discussion, since---in contrast 
to the asymptotically flat $\Lambda=0$ case---the
competition with the  decay  associated to the region
between the event and cosmological horizons is  delicate, especially as
the extremal limit is approached.
Of particular interest is the question
as to whether generic, admissible initial data posed on a Cauchy surface 
lead to solutions whose local (integrated) energy blows up at the Cauchy horizon, for this 
statement holds in the asymptotically
flat case and would correspond precisely 
to the blow up required by Christodoulou's formulation of strong cosmic censorship. 
Some recent heuristic work suggests that the answer is in general negative for solutions
arising from sufficiently
smooth
data, i.e.~there exists
a certain range of black hole parameters such that 
for  all such data,
the arising solutions have finite
local (integrated) energy at the Cauchy horizon.
In this short note, we shall show in contrast that, by slightly relaxing the smoothness
assumption on initial data, we are able to prove the analogue of the Christodoulou
statement in the affirmative,
 i.e.~we show that for generic  data in our allowed class,
 the local energy blow-up statement indeed holds at the Cauchy horizon,
for all subextremal black hole parameter ranges. 
We present two distinct proofs. The first is based on an explicit
mode construction while the other  is softer and uses only time translation invariance
of appropriate scattering maps, in analogy with our previous
[\emph{Time-translation invariance of scattering maps and blue-shift instabilities on Kerr black hole spacetimes},
 Commun.~Math.~Phys.~350 (2017), 985--1016].
 Both proofs use statements concerning the non-triviality of transmission and
reflexion, which are easy to infer by o.d.e.~techniques and analyticity considerations.
  Our slightly enlarged class of initial data is 
 still sufficiently regular to ensure both stability and decay properties in the
region between the event and cosmological horizons as well as
the boundedness and 
continuous extendibility beyond the Cauchy horizon. This suggests thus that it is finally this 
class---and not smoother data---which may provide the  correct setting to formulate 
the genericity condition in  strong cosmic censorship.
\end{abstract}

\section{Introduction}
\label{theintrosec}

Penrose was the first to remark that the
\emph{Cauchy horizon} of the Reissner--Nordstr\"om and Kerr black hole solutions
 is subject to
a \emph{blue-shift instability}~\cite{penrose1968battelle}. 
As the Cauchy  horizon
delimits the region of spacetime
determined uniquely by initial data, this instability provided a path to a possible resolution to one
of general relativity's most puzzling paradoxes:  Observers crossing the Cauchy horizon
experience failure of predictability in a supposedly deterministic classical-physics 
theory, without however
manifestly exiting the domain of validity of the classical description.  The blue-shift instability
associated to the Cauchy horizon  eventually led to the formulation 
of the  \emph{strong cosmic censorship conjecture}~\cite{PenroseSCCrefer}, 
according to which, for \emph{generic} 
initial data for the Einstein vacuum equations
\begin{equation}
\label{Einsteinvac}
R_{\mu\nu}[g]=0,
\end{equation}
or more general Einstein matter systems like the Einstein--Maxwell equations,
the spacetime region $(\mathcal{M},g)$ uniquely determined by initial data (i.e.~the ``maximal Cauchy development'') is
suitably inextendible. 
A similar conjecture can be made for the vacuum equations
with  a  cosmological constant:
\begin{equation}
\label{EinsteinvacLambda}
R_{\mu\nu}[g]=\Lambda g_{\mu\nu},
\end{equation}
 which admit in the case $\Lambda>0$
the so-called  Kerr--de Sitter solutions, or
for the more general Einstein--Maxwell equations with $\Lambda>0$
\begin{align}
 	\label{eq:einsteinmaxwell}
	&R_{\mu\nu}[g] = \Lambda g_{\mu\nu}+ 2 \left( F_{\mu}^{~\lambda} F_{\lambda\nu}  - \frac{1}{4} g_{\mu\nu} F_{\lambda \kappa} F^{\lambda \kappa} \right), \qquad 
\nabla^\nu F_{\mu\nu} = 0 , \qquad \nabla_{[\mu} F_{\nu\lambda]} = 0,
\end{align}
which admit the 
Reissner--Nordstr\"om--de Sitter solutions
(in fact the Kerr--Newman--de Sitter solutions encompassing
all of the previous),
all again possessing Cauchy horizons inside of black holes.
 For the conjecture  to be made precise, one must in particular stipulate 
\emph{in what sense} $\mathcal{M}$ should be inextendible. The stronger the inextendibility condition,
the more definitive a resolution to the paradox the conjecture would provide.
The most satisfying statement would be if  the  spacetime metric $g$
itself could be shown to
be  generically inextendible merely as a continuous ($C^0$) Lorentzian
metric~\cite{Chrmil},
i.e.~without requiring further differentiability. This is the so-called
\emph{$C^0$-formulation} of strong cosmic censorship.
This formulation would correspond to the inextendibility statement which
indeed holds for Schwarzschild~\cite{SbierskiCzero} across its $r=0$ singularity
and
 is related to the property that  observers approaching $r=0$ 
 are ``torn apart'' by infinite tidal deformations, 
making the issue of  their future (as classical observers)
a moot point.
 The question of which (if any)  formulation of strong cosmic censorship holds hinges in turn  
on the \emph{strength} of the blue-shift instability.

A proxy problem for understanding the above issue for $(\ref{Einsteinvac})$ is 
to consider just the linear wave equation
\begin{equation}
\label{linearwaveequation}
\Box_g\psi=0
\end{equation}
on a fixed Reissner--Nordstr\"om or Kerr background.
Remarkably, it has been shown that, for solutions to $(\ref{linearwaveequation})$ arising from
sufficiently regular and localised initial data posed on a Cauchy hypersurface,
$\psi$ remains uniformly bounded on the entire maximal Cauchy development, in particular, on
the  black hole interior, up to and including the Cauchy horizon,
to which $\psi$ in fact extends continuously~\cite{annefranzen, Franzen2}.  
See also~\cite{Hintz:2015koq}.
Thus,  the \emph{amplitude} of $\psi$ is not affected by any blue-shift instability,
which only acts on \emph{derivatives} of $\psi$.
Indeed, for solutions arising from
\emph{generic} initial data in the above class, the  derivative of $\psi$ transversal 
to the Cauchy horizon has been proven~\cite{LukOhpub, DafShlap} to blow up identically
along the Cauchy horizon, in fact, $\psi$ fails to be in 
the Sobolev space $H^1_{\rm loc}$. This  means that extensions
of $\psi$ not only fail to solve $(\ref{linearwaveequation})$ classically
at the Cauchy horizon
 but cannot there
 be interpreted as ``finite-energy'' weak solutions of $(\ref{linearwaveequation})$.
 (For a previous conditional instability results, see~\cite{mcnamara1978instability, D2}.)
To obtain the above statement, it was essential 
 that the natural localisation assumption on data implies that solutions generically
 decay only inverse polynomially on the event horizon, which is then
dwarfed by  the
blue-shift at the Cauchy horizon, governed by an exponential growth mechanism. 
The proof given in our~\cite{DafShlap} was quite soft, exploiting  directly
the translation invariance and non-trivial transmission properties of the scattering
map, together with the properties of the Killing  generator of the 
Cauchy horizon.  An alternative approach has been given in~\cite{LukSbierski}, directly relating a lower bound on the event horizon to the blow up statement on the Cauchy horizon.

While Penrose's blue-shift instability property is familiar to the wider relativity community, 
the question of its precise strength has remained mostly confined to more specialist literature. 
Indeed,
the amplitude stability and continuous extendibility result for $(\ref{linearwaveequation})$
referred to above,
first suggested  by~\cite{McNamara1} and proven finally in~\cite{annefranzen},
remained largely unknown---perhaps because 
it was originally thought that one could not extrapolate
these latter stability statements to the fully nonlinear equations $(\ref{Einsteinvac})$. Indeed,
it was widely expected that the quadratic terms in $(\ref{Einsteinvac})$ would
lead to a highly non-linear behaviour once derivatives were sufficiently large,
leading to a breakdown of the basic causal structure of the spacetime metric $g$
associated to the Cauchy horizon, forming instead
a spacelike singularity beyond
which the metric $g$ itself failed to be continuously extendible, just as in Schwarzschild.
This expectation has been definitively
falsified, however, in~\cite{DafLuk1}, where it has been proven, in the context
of the fully nonlinear evolution under $(\ref{Einsteinvac})$, without symmetry assumptions,
that---assuming only the stability of the 
Kerr exterior region---then it follows that the Cauchy horizon persists as a null 
boundary of spacetime, and the metric indeed extends beyond continuously. 
(See~\cite{Hiscock, PI1, ori1991inner, D2} 
for earlier work on model spherically symmetric problems
and~\cite{Ori1997} for an heuristic study of the problem without symmetry.)
Thus, satisfying though it would have been, the
$C^0$-formulation of strong cosmic censorship,
described at the beginning of this paper, is in fact \underline{false}!
This  motivated  Christodoulou's  reformulation of strong cosmic censorship~\cite{Chr},
where ``inextendibility'' is stated in the class of metrics not just merely continuous but
now also required to
have locally square integrable  Christoffel symbols. 
In  analogy with the proxy problem $(\ref{linearwaveequation})$, this
formulation corresponds precisely to blow up in
$H^1_{\rm loc}$, and represents the threshold for the standard notion of weak
solution of the Einstein equations $(\ref{Einsteinvac})$. Though this notion of
inextendibility is not sufficient to ensure that classical
observers are torn apart by infinite tidal deformations before exiting
the domain of predictability, it still
provides a definitive sense in which the classical-physics description 
can be said to   locally break down
whenever predictability fails,
giving thus at least a partially 
satisfactory resolution to the paradox of 
Cauchy horizons.\footnote{There of course
is an even weaker formulation of strong cosmic censorship, where inextendibilty 
is required in the sense of a $C^2$ Lorentzian metric. We will take
the point of view here, however, that the $C^2$-formulation is manifestly 
unsatisfactory,
given that strong local well posedness results have already been 
shown for $(\ref{Einsteinvac})$ well below the $C^2$ threshold,
for instance at the level
of data which are $H^2$~\cite{Klainerman2015}. This point of view
is nicely explained also in~\cite{cardosoetal}.
Nonetheless, the $C^2$ formulation is still very useful to consider as a test-case for what
can be proven!}
It remains an open problem 
to show that the Christodoulou formulation indeed holds, at the very least in a neighbourhood of Kerr,
but there has been supporting recent work on spherically symmetric 
model 
problems~\cite{D2, LukOh2017one, LukOh2017two, VandeMoortel2018}
and for the vacuum without symmetry in~\cite{Lukweaknull}. 

Turning to the case $\Lambda>0$, already through the prism of the proxy problem $(\ref{linearwaveequation})$, the issue of the nature
of the blow-up for $\psi$ on the  Kerr--de Sitter and Reissner--Nordstr\"om--de Sitter 
spacetimes (satisfying $(\ref{EinsteinvacLambda})$ and $(\ref{eq:einsteinmaxwell})$ respectively) appears more delicate, and has been mired in confusion. See the discussion
in~\cite{Dafermos2014}. As opposed
to the asymptotically flat spacetimes satisfying $(\ref{Einsteinvac})$, 
where, for generic appropriate data,
$\psi$ decays inverse polynomially on 
the event horizon, in the latter two $\Lambda>0$ cases, $\psi$ decays exponentially~\cite{Haefner, 
Dafermos:2007jd, vasy}. 
On the one hand, this very fast decay is extremely fortuitous for proving non linear stability
theorems. Indeed,
exploiting this fast decay,
the stability properties of the region between the event and cosmological horizon,
under the full  nonlinear evolution of $(\ref{EinsteinvacLambda})$ and more generally 
$(\ref{eq:einsteinmaxwell})$,
have been inferred in the very slowly rotating case $|a|\ll M, Q$ in the remarkable 
recent~\cite{hintz2016global, hintz2016non}.\footnote{In particular, this allows
one to unconditionally apply the analogue of~\cite{DafLuk1} in the black hole interior
to definitively falsify the $C^0$ formulation of strong cosmic censorship for
$(\ref{EinsteinvacLambda})$ or more generally $(\ref{eq:einsteinmaxwell})$ in the
$\Lambda>0$ case. For work on a spherically symmetric non-linear toy model problem for
understanding strong cosmic censorship with $\Lambda>0$, see~\cite{Costa2017, possphsymmint}
and references therein. For the linear wave equation $(\ref{linearwaveequation})$ in the
black hole interior in the $\Lambda>0$ case see~\cite{hintzvasycosmoint}.}
This exponential decay means 
in principle, however, that decay along the event horizon is  in direct competition 
with the blue-shift associated to the Cauchy horizon generating exponential growth.
The precise exponential rates now matter!

For smooth initial data, the asymptotic behaviour of solutions in the
region between the event and cosmological horizons should be governed
by quasinormal modes~\cite{Chandrasekhar, dyatlov2011quasi,warnick}. Thus, the question of whether the blue-shift wins
appears to be connected to determining the ``spectral gap'',
the infimum of the imaginary parts of the quasinormal modes. 
For this, one must take into account phenomena connected to 
the event and cosmological horizons, 
slowly decaying solutions corresponding to trapped null geodesics,
as well as other, more subtle slowly damped modes corresponding to the near-extremal limit in
the superradiant case. The relevance of these for the problem at hand has been 
discussed in~\cite{bradymossmyers},~\cite{cardosoetal} and~\cite{Dias:2018ynt} 
respectively. (See also~\cite{hodhod} for the case of a spherically symmetric charged scalar field.) Remarkably, in the Kerr--de Sitter  case,  
 for all subextremal 
parameters,
it has very recently been argued that the spectral
gap is necessarily sufficiently small so as to expect the blue-shift effect to be
strong enough so as for 
the $H^1_{\rm loc}$ blow-up result to still hold~\cite{Dias:2018ynt}.
In the Reissner--Nordstr\"om--de Sitter case, however, there remains a range of black hole 
parameters for which the prospect of a relatively large
 spectral gap ``survives'' all the above obstructions, suggesting in particular
 that solutions $(\ref{linearwaveequation})$ do now extend to be $H^1_{\rm loc}$ at the
Cauchy horizon~\cite{cardosoetal, Dias:2018ynt}.
The above discussion thus  suggests the intriguing possibility that when
$\Lambda>0$, Christodoulou's formulation
of strong cosmic censorship is violated~\cite{Dafermos2014, HarveyPhy,cardosoetal} (at least
 for $(\ref{eq:einsteinmaxwell})$
if not for  $(\ref{EinsteinvacLambda})$).

The above apparent failure of even Christodoulou's revised formulation of strong cosmic censorship (already a weakening of the original $C^0$-formulation!)~would leave
a rather discomforting situation for general relativity in the presence of a positive
cosmological constant $\Lambda>0$: 
For if Cauchy horizons generically occur at which
spacetime can moreover still be interpreted as a weak solution of
the Einstein equations, then it is difficult to argue decisively that the classical description has
``broken down'', and thus, it would appear that the paradox  persists 
of classical predictability failing without manifestly exiting the classical regime.

{\bf \emph{The purpose of this short note is to suggest a way out.}}
We will prove  that, at the level of the proxy
 problem $(\ref{linearwaveequation})$,  there is indeed a way to retain the
 desirable generic
$H^1_{\rm loc}$ blowup at the Cauchy horizon: \emph{It suffices to 
consider a slightly less regular, but still well-motivated, class of initial data.}

To formulate our result, let $\widetilde{\mathcal{M}}$ denote maximally extended subextremal
Reissner--Nordstr\"om--de Sitter spacetime (or more generally,  Kerr--Newman--de Sitter
spacetime). Let $\widetilde\Sigma$ denote a complete spacelike hypersurface intersecting
two  cosmological horizons $\mathcal{C}^+$ as in Figure~\ref{fig1}. 
Initial data $(\Psi, \Psi')$ on $\widetilde\Sigma$ give rise to a solution $\psi$
on the future domain of dependence $D^+(\widetilde\Sigma)$, with $\psi|_{\widetilde{\Sigma}}=
\Psi$, $n_{\widetilde{\Sigma}}\psi|_{\widetilde{\Sigma}}=\Psi'$, where $n_{\widetilde{\Sigma}}$
denotes the future normal to $\widetilde{\Sigma}$.  
The local energy flux of $\psi$ along $\widetilde\Sigma$ is of course computable
in terms of initial data $(\Psi, \Psi')$, in particular $\psi$ has finite local energy flux along
$\widetilde\Sigma$ if $(\Psi,\Psi')\in H^1_{\rm loc}(\widetilde\Sigma)\times
L^2_{\rm loc}(\widetilde\Sigma)$. For brevity, we will say in this case
that the data $(\Psi, \Psi')$ have finite local energy
along $\widetilde{\Sigma}$. 

Our main result is the following
\begin{theorem}
\label{maintheoremINTRO}
Consider a subextremal Reissner--Nordstr\"om--de Sitter spacetime, or more generally,
Kerr--Newman--de Sitter spacetime $\widetilde{\mathcal{M}}$. 
For generic initial data $(\Psi, \Psi')$
with finite local energy along 
$\widetilde{\Sigma}$,  
the resulting solution $\psi$ of $(\ref{linearwaveequation})$
in $D^+(\widetilde{\Sigma})$ 
has infinite local energy along
hypersurfaces intersecting transversally the Cauchy horizon $\mathcal{CH}^+$,
i.e.~$\psi$ in particular fails to extend $H^1_{\rm loc}$ around any point
of $\mathcal{CH}^+$.
\end{theorem}

\begin{figure}
\begin{center}
\begin{tikzpicture}
\definecolor{light-gray}{gray}{.9}
\fill[light-gray] (-2,2)--(0,0) to [out = -20,in=180]  (2,-.4) to [out = 0,in=210] (3,-.2) -- (3,-3.8);
\fill[light-gray] (-2,2)--(0,-4) to [out = 20,in = 180] (2,-3.6) to [out = 0, in=150] (3,-3.8);
\fill[light-gray] (-2,2) -- (0,-4) -- (-2,-6) -- (-4,-4) to [out = 160, in = 0] (-6,-3.6) to [out = 0,in = 30] (-7,-3.8);
\fill[light-gray] (-2,2) -- (-4,0) to [out = 200,in = 0] (-6,-.4) to [out = 0,in = -30] (-7,-.2) --(-7,-3.8);

\draw (0,0) -- (2,-2) node[sloped,above,midway]{$\mathcal{C}^+$}; 
\draw (2,-2) --  (0,-4) node[sloped,below,midway]{$\mathcal{C}^-$};
\draw (0,-4) -- (-2,-2) node[sloped,below,midway]{$\mathcal{H}^-$}; 
\draw (-2,-2) -- (-4,0) node[sloped,above,midway]{$\mathcal{H}^+$}; 
\draw (-2,-2) -- (0,0) node[sloped,above,midway]{$\mathcal{H}^+$}; 
\draw (-4,0) -- (-2,2) node[sloped,above,midway]{$\mathcal{CH}^+$}; 
\draw (-2,2) -- (0,0) node[sloped,above,midway]{$\mathcal{CH}^+$}; 


\draw (-4,0) -- (-2,-2); 
\draw (-2,-2) --  (-4,-4) node[sloped,below,midway]{$\mathcal{H}^-$};
\draw (-4,-4) -- (-2,-6) node[sloped,below,midway]{$\mathcal{CH}^-$};
\draw (-2,-6) -- (0,-4) node[sloped,below,midway]{$\mathcal{CH}^-$};
\draw (-4,-4) -- (-6,-2);  
\draw (-6,-2) -- (-4,0);
\draw (0,-4) -- (-2,-6) -- (-4,-4);

\draw [dashed] (0,0) to [out = -20,in=180]  (2,-.4) to [out = 0,in=210] (3,-.2);

\draw [dashed] (0,-4) to [out = 20,in = 180] (2,-3.6) to [out = 0, in=150] (3,-3.8);
\draw (2,-2) -- (3,-1); 
\draw (-6,-2) -- (-4,0) node[above,midway]{$\mathcal{C}^+$};
\draw (-6,-2) -- (-4,-4) node[below,midway]{$\mathcal{C}^-$};
\draw (2,-2) -- (3,-3);

\draw [dashed] (-4,0) to [out = 200,in = 0] (-6,-.4) to [out = 0,in = -30] (-7,-.2) ;
\draw [dashed] (-4,-4) to [out = 160, in = 0] (-6,-3.6) to [out = 0,in = 30] (-7,-3.8);
\draw (-6,-2) -- (-7,-1); 
\draw (-6,-2) -- (-7,-3);

\path [draw=black,fill=white] (0,0) circle (1/16); 
\path [draw=black,fill=black] (2,-2) circle (1/16);
\path [draw=black,fill=white] (0,-4) circle (1/16); 
\path [draw=black,fill=white] (-4,-4) circle (1/16); 
\path [draw=black,fill=black] (-2,-2) circle (1/16);
\path [draw=black,fill=white] (-4,0) circle (1/16); 
\path [draw=black,fill=black] (-2,2) circle (1/16);
\path [draw = black, fill = black] (-2,-6) circle(1/16);
\path [draw = black,fill = black] (-6,-2) circle (1/16);
\draw (3,-2) node[right]{$\ldots$};
\draw (-7,-2) node[left]{$\ldots$};

\draw[very thick] (3,-.5) to [out=200,in=-20]node[sloped,above,midway]{$\widetilde{\Sigma}$} (-7,-.5) ;

\draw (2,-.1) node{$\mathcal{I}^+$};
\draw (2,-3.9) node{$\mathcal{I}^-$};
\draw (-6,-.1) node {$\mathcal{I}^+$};
\draw (-6,-3.9) node {$\mathcal{I}^-$};

\end{tikzpicture}
\end{center}
\caption{Portion of maximally extended Reissner--Nordstr\"{o}m--de Sitter and a hypersurface $\widetilde{\Sigma}$}\label{fig1}
\end{figure}
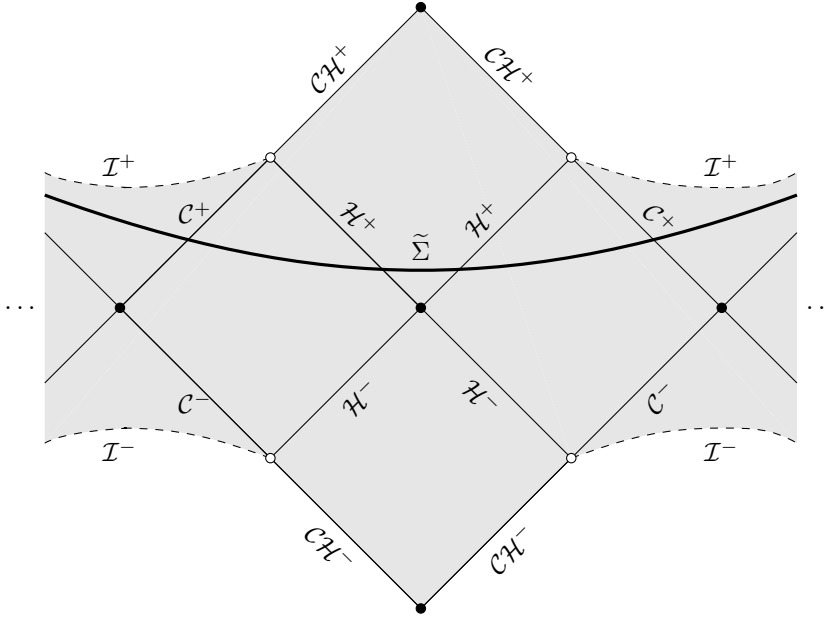

The genericity statement can be understood as the following ``co-dimension $1$ property'': 
For all Cauchy data $\left(\Psi_0,\Psi'_0\right)$ which lead to a solution $\psi_0$ of finite energy along hypersurfaces transversally intersecting the Cauchy horizon $\mathcal{CH}^+$, the solution $\psi$ corresponding to the Cauchy data $\left(\Psi_0 + c\Psi_1,\Psi'_0 + c\Psi'_1\right)$ has infinite energy along hypersurfaces transversally intersecting $\mathcal{CH}^+$ for some $\left(\Psi_1,\Psi_1'\right)$ and every $c\in \mathbb{R}\setminus\{0\}$. By linearity, it suffices to construct a single $\left(\Psi_1,\Psi_1'\right)$ in the case  $\left(\Psi_0,\Psi_0'\right) = (0,0)$. Note that this is analogous to the notion of genericity used by Christodoulou in his proof of weak cosmic censorship for the spherically symmetric Einstein-scalar-field system~\cite{Christodoulou4,Chrmil}. (We also observe that one can show that the initial data leading to the desired Cauchy horizon blow-up form a set of Baire second category within the class of all initial data. However, since smooth solutions are dense in $H^1$, in light of~\cite{Dafermos2014, HarveyPhy,cardosoetal} referred to above, we do not expect to show that the set of Cauchy data leading to $H^1_{\rm loc}$-blow up at $\mathcal{CH}^+$ is open.)

In fact, one can take generic initial data $(\Psi, \Psi')$ in a slightly more regular class,
i.e.~one can replace the assumption that  $\psi$ merely has finite local energy on $\widetilde{\Sigma}$ (i.e.~corresponding to data lying in the Sobolev space
$H^1_{\rm loc}(\widetilde{\Sigma})\times L^2_{\rm loc}(\widetilde{\Sigma})$), 
with the assumption that the data lie in
$H^{1+\epsilon}_{\rm loc}(\widetilde{\Sigma})\times H^\epsilon_{\rm loc}(\widetilde{\Sigma})$,
with $\epsilon\to 0$ however as extremality is approached.

Recall from the above discussion that, by linearity,  to obtain the above theorem it suffices to produce
a \emph{single} solution $\psi$, arising from data in the admissible  space, 
satisfying the claimed
blow up.
In the Reissner--Nordstr\"om--de Sitter case,
we can  in fact construct a spherically symmetric such $\psi$, or more generally,
 a $\psi$ whose angular frequency is 
supported on an arbitrary
fixed spherical harmonic number $\ell$. 
This means we can replace the space $H^{1(+\epsilon)}_{\rm loc}(\widetilde{\Sigma})\times H^{0(+\epsilon)}_{\rm loc}(\widetilde{\Sigma})$ in the above
statement with
a space (let us call it $\mathcal{D}$)
with \emph{arbitrary additional regularity in the angular directions}. Solutions 
$\psi$ of $(\ref{linearwaveequation})$ arising from  data in $\mathcal{D}$ would then
share the same phenomenology of behaviour as smooth $C^\infty$ solutions
(in particular, in view of the norms of~\cite{Dafermos:2007jd}, the analogue of the
results of~\cite{annefranzen} still apply). 
Moreover, the fact that one formulates strong cosmic censorship in terms
of inextendibility with low regularity strongly suggests that one should also
allow similarly low regularity initial data (cf.~the role of low regularity in 
the genericity assumption in the proof of weak cosmic censorship
under spherical symmetry~\cite{Christodoulou4}).
One can thus argue that there is no particular reason to prefer smoother initial data, and perhaps 
the above class of data indeed provides a more appropriate
setting in which to consider the genericity assumption of strong cosmic censorship.
We shall discuss this further in  Section~\ref{discussionsec}.

We shall here carry out the proof of Theorem~\ref{maintheoremINTRO}
in detail only in the Reissner--Nordstr\"om--de Sitter case,
as already in this case, the $H^1_{\rm loc}$ blowup shown here is thought
to fail for smooth data if the black hole parameters are suitably
close to extremality~\cite{Dafermos2014, cardosoetal}. 
To make this paper self-contained, we will give  
an explict construction of  the  Reissner--Nordstr\"om--de Sitter 
metric in the relevant region $\mathcal{M}$ (see Section~\ref{prelim}).
The theorem  easily reduces to
Theorem~\ref{mainresult} (see Section~\ref{resultsection}), 
concerning only region $\mathcal{M}$, 
 which is the precise formulation we shall prove. 
We shall in fact
provide \emph{two} distinct proofs.
Our first proof (see Section~\ref{modes}) is explicit and constructs a solution suitably 
blowing up at the Cauchy horizon as a mode solution
whose time frequency has negative imaginary part,  related to the regularity
at the event horizon. We note that this imaginary part may be \emph{less than} the
spectral gap associated to smooth data; thus one sees how the problem posed
by determining the precise spectral gap is completely circumvented by passing to lower 
regularity (cf.~Definition 3.19 given in~\cite{warnick} for an $H^k$-quasinormal mode
and also the example discussed in Section~6 of that paper).
Our second proof (see Section~\ref{timetranslate})
is softer, and appeals directly to the time translation invariance properties of scattering
maps.
In both proofs, the fundamental inequality $\kappa_->\kappa_+$ connecting
the surface gravities of the Cauchy and event horizons plays an essential role.
In addition, both proofs require appealing to the nonvanishing of transmission and reflexion of
suitable scattering maps for a certain  open set  of frequencies; this is here inferred
exploiting analyticity properties.
As is clear from the second proof, 
equation $(\ref{linearwaveequation})$ can be replaced by 
a wide class of translation-invariant wave-type
equations, on spacetimes sharing only the basic qualitative properties of the 
respective horizons. 
We shall leave, however, such further generalisations of our result to another 
occasion.

\paragraph{Acknowledgement.} The authors would like to thank Christoph Kehle,
Harvey Reall, 
Igor Rodnianski, Jorge Santos and Claude Warnick for 
useful conversations and sharing their insights on aspects of this problem. 
YS acknowledges support from the NSF Postdoctoral Research Fellowship under award no.\ 1502569. MD acknowledges support through NSF grant DMS-1709270 and EPSRC grant EP/K00865X/1.

\section{The Reissner--Nordstr\"om--de Sitter metrics}\label{prelim}
In this section we will quickly review the structure of the Reissner--Nordstr\"om--de 
Sitter spacetime and introduce the relevant notation.

We will eventually define a manifold $(\mathcal{M},g)$ with stratified boundary corresponding
to the union of a static region $\mathcal{M}_{\rm static}$ 
(bounded by a bifurcate event horizon $\mathcal{H}^+_A\cup\mathcal{H}^-$ 
and cosmological horizon $\mathcal{C}^+$) and
a black hole interior region
$\mathcal{M}_{\rm int}$ 
(bounded by a bifurcate event horizon $\mathcal{H}^+_B\cup\mathcal{H}^+_A$ 
and Cauchy horizon $\mathcal{CH}^+_B\cup \mathcal{CH}^+_A$).
See already Figure~\ref{fig2}. 
We review this construction explicitly here, as we shall make use of the properties
of the various underlying coordinate systems. 

Note that the above  $(\mathcal{M},g)$ 
is itself still only a subset of the maximally extended $(\widetilde{\mathcal{M}} , \widetilde{g})$
Reissner--Nordstr\"om--de Sitter referred to in Section~\ref{theintrosec}.
Since Theorem~\ref{maintheoremINTRO} quickly reduces to a
statement on $(\mathcal{M},g)$, we shall not discuss the explict construction of
$\widetilde{\mathcal{M}}$ here (see already Section~\ref{resultsection}).

\subsection{Schwarzschild coordinates and the static patch $\mathcal{M}_{\rm static}$}
We say that a three-tuple of positive constants $(M,e,\Lambda)$ is \emph{non-degenerate} if the function
\[1 - \frac{2M}{r} +\frac{e^2}{r^2} - \frac{\Lambda}{3}r^2\]
has three distinct positive zeros, which we then label
\begin{equation}\label{rplusminusc}
0 < r_- < r_+ < r_c < \infty.
\end{equation}
Henceforth, we shall always consider a fixed choice of such $(M,e,\Lambda)$.

We define the  \emph{static region} of  
Reissner--Nordstr\"{o}m--de Sitter  with parameters 
 $(M,e,\Lambda)$
 to be the manifold  $\mathcal{M}_{{\rm static}}$ defined  by the coordinate range 
 $(t,r,\theta,\phi) \in \mathbb{R} \times (r_+,r_c) \times \mathbb{S}^2 \doteq \mathcal{M}_{{\rm static}}$ 
 with  metric given by
\begin{equation}\label{metricsimple}
g \doteq -\left(1 - \frac{2M}{r} +\frac{e^2}{r^2} - \frac{\Lambda}{3}r^2\right)dt^2 + \left(1 - \frac{2M}{r} +\frac{e^2}{r^2} - \frac{\Lambda}{3}r^2\right)^{-1}dr^2 + r^2d\sigma_{\mathbb{S}^2},
\end{equation}
where $d\sigma_{\mathbb{S}^2}=d\theta^2+\sin^2\theta d\phi^2$ 
denotes the round metric on $\mathbb{S}^2$.  
We call the above coordinates \emph{Schwarzschild coordinates}.
Note that the metric $(\ref{metricsimple})$ is manifestly stationary and spherically
symmetric, with Killing fields $T\doteq \partial_t$ and $\Omega_1,\Omega_2, \Omega_3$,
where $\Omega_i$ denote the standard angular momentum operators in $(\theta,\phi)$
coordinates.

There exists  a one parameter family
of stationary spherically symmetric two-forms $F_{\mu\nu}$
on $\mathcal{M}_{{\rm static}}$
 so that the triple $(\mathcal{M}_{{\rm static}},g, F)$ is now a solution to the Einstein--Maxwell equations with a positive cosmological constant $\Lambda>0$ $(\ref{eq:einsteinmaxwell})$ (see~\cite{carterreview}).
The choice of electromagnetic field will have no relevance in this paper; 
in what follows, we shall only refer to the underlying metric $(\ref{metricsimple})$.

We will time orient $\left(\mathcal{M}_{{\rm static}},g\right)$ with the timelike vector field $T=\partial_t$. 

\subsection{Outgoing Eddington--Finkelstein coordinates attaching $\mathcal{CH}^+_B\cup\mathcal{M}_{\rm interior}\cup \mathcal{H}^+_A$ and $\mathcal{C}^-$}
\label{outgoingcorfir}

The metric $g$ defined by~\eqref{metricsimple} can be smoothly extended to a larger manifold; we now succinctly review the construction. 

First, it is convenient to introduce a function $r^*(r)$ by setting
\begin{equation}\label{rstar}
\frac{dr^*}{dr} = \left(1 - \frac{2M}{r} +\frac{e^2}{r^2} - \frac{\Lambda}{3}r^2\right)^{-1},\qquad r^*\left(\frac{r_++r_c}{2}\right) = 0.
\end{equation}
Note that the range $r \in (r_+,r_c)$ corresponds to $r^* \in (-\infty,\infty)$. Then we may define $v \doteq t+ r^*$, and one finds that in the \emph{outgoing Eddington--Finklestein coordinates} $(v,r,\theta,\phi)$, the metric $g$ defined by~\eqref{metricsimple} becomes
\begin{equation}\label{metricoutgoing}
g = -\left(1 - \frac{2M}{r} +\frac{e^2}{r^2} - \frac{\Lambda}{3}r^2\right)dv^2 + 2dvdr + r^2d\sigma_{\mathbb{S}^2}.
\end{equation}
The manifold $\mathcal{M}_{{\rm static}}$ corresponds to the coordinate range $(v,r,\theta,\phi) \in \mathbb{R} \times (r_+,r_c) \times \mathbb{S}^2$.

It is now manifest that the expression~\eqref{metricoutgoing} in fact also defines a smooth Lorentzian metric on the manifold with boundary $\mathcal{M}_0$ defined by the coordinate range
$(v,r,\theta,\phi) \in \mathbb{R} \times [r_-,r_c] \times \mathbb{S}^2\doteq \mathcal{M}_0$. 
The manifold $\mathcal{M}_{{\rm static}}$ is thus an open submanifold of
$\mathcal{M}_0$ corresponding to the subset $r_+<r<r_c$. Note that defining
$T=\partial_v$ with respect to the above coordinates, $T$ is a Killing field
on $\mathcal{M}_0$ smoothly extending 
the definition from $\mathcal{M}_{{\rm static}}$. 

Let us define the \emph{black hole interior region} to be the open subset
$\mathcal{M}_{{\rm interior}}\doteq \{p\in\mathcal{M}_0:r_-<r(p)<r_+\}$ 
and the \emph{outgoing future event horizon} to
be the hypersurface
$\mathcal{H}^+_A= \{p\in \mathcal{M}_0:r(p)=r_+\}$.  Note that $T$ is spacelike on
$\mathcal{M}_{{\rm interior}}$ and null and tangent to $\mathcal{H}^+_A$.
In particular, $\mathcal{H}^+_A$ is a null hypersurface. 

The boundary of $\mathcal{M}_0$ (as a manifold with boundary) consists
of two components $\partial\mathcal{M}_0= \mathcal{CH}^+_B\cup \mathcal{C}^-$ defined by 
\[
\mathcal{CH}^+_B \doteq \{p\in \mathcal{M}_0: r(p) = r_-\} ,\qquad \mathcal{C}^-
 \doteq \{p\in \mathcal{M}_0 :r(p) = r_c\}.
\]
The vector field $T$ is null and tangential to these hypersurfaces which
are in particular thus null and Killing horizons. We shall call $\mathcal{CH}^+_B$ the \emph{outgoing
Cauchy horizon} and $\mathcal{C}^-$ the  \emph{past cosmological horizon}.

Finally let us note that the future orientation defined by $T$ 
on $\mathcal{M}_{\rm static}$ extends
to a unique future orientation on $\mathcal{M}_0$. According to this,
the vector field $T$ is future directed null on $\mathcal{H}^+_A$, $\mathcal{CH}^+_B$
and $\mathcal{C}^-$.

\subsection{Ingoing Eddington--Finkelstein coordinates on $\mathcal{M}_{\rm static}$  attaching $\mathcal{H}^-$ and $\mathcal{C}^+$}
\label{ingoingfir}

Returning to $\left(\mathcal{M}_{{\rm static}},g\right)$, we can set $u \doteq t- r^*$ and similarly define \emph{ingoing Eddington--Finklestein coordinates} $(u,r,\theta,\phi)$. The manifold $\mathcal{M}_{{\rm static}}$ corresponds to $(u,r,\theta,\phi) \in \mathbb{R} \times (r_+,r_c)\times \mathbb{S}^2$. The metric now takes the form
\begin{equation}\label{metricingoing}
g = -\left(1 - \frac{2M}{r} +\frac{e^2}{r^2} - \frac{\Lambda}{3}r^2\right)du^2 - 2dudr + r^2\sigma_{\mathbb{S}^2}.
\end{equation}

Similarly to the previous section, it is immediately clear that the metric $g$ extends to a smooth Lorentzian manifold with boundary $\widehat{\mathcal{M}}_{\rm static}$
defined by the coordinate range 
$(u,r,\theta,\phi) \in \mathbb{R} \times [r_+,r_c] \times \mathbb{S}^2\doteq
\widehat{\mathcal{M}}_{\rm static}$. 
The boundary $\partial \widehat{\mathcal{M}}_{\rm static}=\mathcal{H}^-\cup \mathcal{C}^+$ 
of this manifold consists of the hypersurfaces 
\[
\mathcal{H}^- \doteq \{p\in \widehat{\mathcal{M}}_{\rm static}:
r(p)=r_+\}, \qquad \mathcal{C}^+ \doteq \{p\in \widehat{\mathcal{M}}_{\rm static}: r(p) = r_c\}.
\]
which we shall refer to respectively as the \emph{past event horizon} and
the \emph{future cosmological horizon}.
Note that the Killing vector field $T$ of the static region extends 
to this new manifold with boundary by $T=\partial_u$, and this is tangential
and null on $\mathcal{H}^-$ and $\mathcal{C}^+$. These are thus null and
Killing horizons.

We may already now attach the above null hypersurfaces
$\mathcal{H}^-$ and $\mathcal{C}^+$ as
additional boundary to the manifold with boundary $\mathcal{M}_0$
to obtain 
a manifold with boundary
$\mathcal{M}_1=\mathcal{M}_0\cup \mathcal{H}^-\cup\mathcal{C}^+$
on which $T$ is globally defined.
Note that $\mathcal{M}_1$
inherits the time orientation from $\mathcal{M}_0$
and $T$ is future-directed on $\mathcal{H}^-$ and $\mathcal{C}^+$.

\subsection{Ingoing Eddington--Finkelstein coordinates on $\mathcal{M}_{\rm interior}$
attaching $\mathcal{CH}^+_A$ and $\mathcal{H}^+_B$}
\label{ingoingcoorsec}

Next we will define similar ingoing Eddington--Finkelstein coordinates on 
 $\mathcal{M}_{{\rm interior}}$, which allow us to attach two additional null hypersurfaces.

In analogy with~\eqref{rstar}, we define the function $r^*(r)$ in the region 
$\mathcal{M}_{{\rm interior}}$ by
 \begin{equation}\label{rstar2}
\frac{dr^*}{dr} = \left(1 - \frac{2M}{r} +\frac{e^2}{r^2} - \frac{\Lambda}{3}r^2\right)^{-1},\qquad r^*\left(\frac{r_-+r_+}{2}\right) = 0.
\end{equation}
We have $r^*(r_-,r_+) = (-\infty,\infty)$. (Note, however, that $\frac{dr^*}{dr} < 0$.) Now, analogously to the region $\mathcal{M}_{{\rm static}}$, we can define $\tilde u = r^*-t$ in 
$\mathcal{M}_{{\rm interior}}$ and then cover $\mathcal{M}_{{\rm interior}}$ with 
coordinates $(\tilde u,r,\theta,\phi) \in \mathbb{R} \times (r_-,r_+) \times \mathbb{S}^2$.  Just as before, the metric extends smoothly to 
$(\tilde u,r,\theta,\phi) \in \mathbb{R} \times [r_-,r_+] \times \mathbb{S}^2\doteq
\widehat{\mathcal{M}}_{\rm interior}$, and this leads to the definition of the boundary
hypersurfaces
\[
\mathcal{H}^+_B \doteq \{p\in \widehat{\mathcal{M}}_{\rm interior}:r(p)=r_+\},\qquad \mathcal{CH}^+_A \doteq \{p\in \widehat{\mathcal{M}}_{\rm interior}:r(p) = r_-\}
\]
which we may now attach to obtain a manifold with boundary $\mathcal{M}_2=\mathcal{M}_1\cup \mathcal{H}^+_B\cup
\mathcal{CH}^+_A$, which has the additional boundary components
$\mathcal{H}^+_B$ and $\mathcal{CH}^+_A$. We will refer to
these as the \emph{ingoing future event horizon} and the \emph{ingoing Cauchy horizon},
respectively.

This new manifold  $\mathcal{M}_2$ with boundary again inherits a time orientation.
The Killing field $T$ extends globally to $\mathcal{M}_2$ and is again null on
(and tangential to) $\mathcal{H}^+_B$ and $\mathcal{CH}^+_A$. Note
however that $T$ is now \emph{past-directed} on both $\mathcal{H}^+_B$ and $\mathcal{CH}^+_A$.

\subsection{The surface gravities $\kappa_+$, $\kappa_-$ and $\kappa_c$ and the 
inequality $\kappa_->\kappa_+$}

Before further extending the metric to obtain our final $\mathcal{M}$, 
let us discuss further the behaviour of the Killing field $T$ on the horizons.

One can infer immediately from the spherical symmetry
of the metric and the fact that 
$T$ is the Killing null generator of the various Killing horizons $\mathcal{C}^{\pm}$, $\mathcal{H}^+_A$, $\mathcal{H}^-$, $\mathcal{H}^+_B$, $\mathcal{CH}^+_B$ and $\mathcal{CH}^+_A$,
that there exist constants $\kappa_c$, $\kappa_+$ and $\kappa_-$ so that 
\[\nabla_TT|_{\mathcal{H}^+_A} = \kappa_+T,\qquad \nabla_TT|_{\mathcal{C}^+} = \kappa_cT,\qquad \nabla_TT|_{\mathcal{CH}^+_B} = -\kappa_-T,\]
\[\nabla_TT|_{\mathcal{H}^-} = -\kappa_+T,\qquad \nabla_TT|_{\mathcal{H}^+_B} = \kappa_+T,\qquad \nabla_TT|_{\mathcal{C}^-} = -\kappa_cT,\qquad \nabla_TT|_{\mathcal{CH}^+_A} = \kappa_-T.\]

These constants $\kappa_+$, $\kappa_-$ and $\kappa_c$ are the various \emph{surface gravities} of the horizons. In order to calculate the surface gravities it is useful to observe that
\[1 - \frac{2M}{r} + \frac{e^2}{r^2} - \frac{\Lambda}{3}r^2 = -\frac{\Lambda}{3r^2}(r-r_-)(r-r_+)(r-r_c)(r-\tilde{r}_c),\]
where
\[\tilde{r}_c = - r_- - r_+-r_c < 0.\]

Then, using that
\[\kappa_c = -\frac{1}{2}\frac{\partial}{\partial r}\left(1-\frac{2M}{r} + \frac{e^2}{r^2} - \frac{\Lambda}{3}r^2\right)|_{r=r_c},\qquad \kappa_+ = \frac{1}{2}\frac{\partial}{\partial r}\left(1-\frac{2M}{r} + \frac{e^2}{r^2} - \frac{\Lambda}{3}r^2\right)|_{r=r_+},\]
\[\kappa_- = -\frac{1}{2}\frac{\partial}{\partial r}\left(1-\frac{2M}{r} + \frac{e^2}{r^2} - \frac{\Lambda}{3}r^2\right)|_{r=r_-},\]
we obtain
\[\kappa_c = \frac{\Lambda}{6}(r_c-r_+)(r_c-r_-)(r_c-\tilde r_c)r_c^{-2},\qquad \kappa_+ = \frac{\Lambda}{6}(r_+-r_-)(r_c-r_+)(r_+-\tilde r_c)r_+^{-2},\]
\[\kappa_- = \frac{\Lambda}{6}(r_+ -r_-)(r_c-r_-)(r_--\tilde r_c)r_-^{-2}.\]
Thus we see that $\kappa_c$, $\kappa_+$ and $\kappa_-$ are all positive. The following
well-known inequality (which the reader can readily verify) is of fundamental importance for our main results:
\begin{lemma}\label{fundsurfgrav} For each tuple $\left(M,e,\Lambda\right)$ of non-degenerate constants, we have 
\begin{equation*}
\kappa_- > \kappa_+.
\end{equation*}
\end{lemma}


\subsection{Kruskal coordinates  attaching the bifurcation spheres $\mathcal{B}_+$, $\mathcal{B}_-$ and $\mathcal{B}_c$}\label{krusky}

Finally, we shall introduce Kruskal coordinates allowing us to extend our manifold to
include the three bifurcation spheres $\mathcal{B}_c$, $\mathcal{B}_+$, and $\mathcal{B}_-$.

We start with the coordinates which will define $\mathcal{B}_c$. We define two functions $U_c(t,r)$ and $V_c(t,r)$ in $\mathcal{M}_{{\rm static}}$ by
\[
U_c\left(t,r\right) \doteq \exp\left(\kappa_c\left(t-r^*\right)\right),\qquad V_c\left(t,r\right) \doteq -\exp\left(-\kappa_c\left(t+r^*\right)\right).
\]
Then $\mathcal{M}_{{\rm static}}$ corresponds to the range $\{(U_c,V_c,\theta,\phi) \in (0,\infty)\times (-\infty,0) \times \mathbb{S}^2\}$. It then turns out that the metric extends to a smooth Lorentzian metric with stratified boundary on $\{(U_c,V_c,\theta,\phi) \in [0,\infty)\times (-\infty,0] \times \mathbb{S}^2\}$ (the explicit form that the metric takes in these coordinates  will not be relevant for this paper). The hypersurface $\{\{U_c = 0\} \times \{V \in (-\infty,0\} \times \mathbb{S}^2\}$ may be identified with $\mathcal{C}^-$ and the hypersurface $\{\{V_c = 0\} \times \{U \in (0,\infty\} \times \mathbb{S}^2\}$ may be identified with $\mathcal{C}^+$. However, we obtain a new sphere $\mathcal{B}_c \doteq \{\{U_c = 0\} \times \{V_c = 0\} \times \mathbb{S}^2\}$.

The vector field $T$ smoothly extends to $\mathcal{B}_c$ 
where it vanishes, and we have the following formulae:
\[\frac{\partial}{\partial U_c}\Big|_{\mathcal{C}^+} = e^{-\kappa_cu}T,\qquad \frac{\partial}{\partial V_c}\Big|_{\mathcal{C}^-} = e^{\kappa_cv}T.\]

The bifurcation spheres $\mathcal{B}_+$ and $\mathcal{B}_-$ are defined in an analogous fashion and end up corresponding to the common boundary of $\mathcal{H}^+_A$ and $\mathcal{H}^-$ and $\mathcal{CH}^+_B$ and $\mathcal{CH}^+_A$ respectively. The associated Kruskal coordinates can be constructed from 
\begin{equation}\label{kruskalplus}
U_+\left(t,r\right)|_{\mathcal{M}_{{\rm static}}} \doteq -\exp\left(-\kappa_+\left(t-r^*\right)\right),\qquad V_+\left(t,r\right)|_{\mathcal{M}_{{\rm static}}} \doteq \exp\left(\kappa_+\left(t+r^*\right)\right),
\end{equation}
\begin{equation}\label{kruskalminus}
U_-\left(t,r\right)|_{\mathcal{M}_{{\rm interior}}} \doteq -\exp\left(-\kappa_-\left(r^*-t\right)\right),\qquad V_-\left(t,r\right)|_{\mathcal{M}_{{\rm interior}}} \doteq -\exp\left(-\kappa_-\left(t+r^*\right)\right).
\end{equation}
As with $\mathcal{B}_c$, $T$ extends smoothly to $\mathcal{B}_+$ and $\mathcal{B}_-$ where it vanishes and we have the formulas:
\begin{equation}\label{Tform}
\frac{\partial}{\partial U_+}|_{\mathcal{H}^-} = e^{\kappa_+u}T,\qquad \frac{\partial}{\partial V_+}|_{\mathcal{H}^+_A} = e^{-\kappa_+v}T,\qquad \frac{\partial}{\partial V_-}|_{\mathcal{CH}^+_B} = e^{\kappa_-v}T,\qquad \frac{\partial}{\partial U_+}|_{\mathcal{CH}^+_A} = e^{\kappa_cu}T.
\end{equation}

We finally define
\begin{equation}
\label{thisisM}
\mathcal{M}=\mathcal{M}_2 \cup\mathcal{B}_-\cup\mathcal{B}_c\cup\mathcal{B}_+
\end{equation}
with differential structure defined by the 
above charts. It follows that
the above is a manifold with stratified boundary on which
 $g$ extends smoothly as a time-oriented Lorentzian metric.
Note that the interior of $\mathcal{M}$ is given by $\mathcal{M}_{\rm interior}\cup
\mathcal{H}^+_A\cup\mathcal{M}_{\rm static}$ and the boundary 
$\partial\mathcal{M}$ is given by the union
\begin{equation}
\label{dis}
\left(\mathcal{H}^+_B\cup\mathcal{B}_+\cup\mathcal{H}^-\right)\bigcup
\left(\mathcal{CH}^+_B\cup\mathcal{B}_-\cup\mathcal{CH}^+_A\right)\bigcup
\left(\mathcal{C}^+\cup\mathcal{B}_c\cup\mathcal{C}^-\right).
\end{equation}
We note that $\left(\mathcal{H}^+_B\cup\mathcal{B}_+\cup\mathcal{H}^-\right)$
is a smooth boundary hypersurface, but the other two sets in brackets in
$(\ref{dis})$ are unions of transversally intersecting smooth hypersurfaces-with-boundary, with
common boundary $\mathcal{B}_-$, $\mathcal{B}_c$ respectively,
e.g.~
\[
\left(\mathcal{CH}^+_B\cup\mathcal{B}_-\cup\mathcal{CH}^+_A\right)=
(\mathcal{CH}^+_B\cup\mathcal{B}_-)\cup 
(\mathcal{B}_-\cup\mathcal{CH}^+_A)
\]
with $\mathcal{CH}^+_B\cap\mathcal{CH}^+_A=\emptyset$.

In Figure~\ref{fig2} we depict the standard Penrose diagram for  $\left(\mathcal{M},g\right)$.

We note finally that the electromagnetic tensor $F_{\mu\nu}$ on $\mathcal{M}_{\rm static}$
extends
to an electromagnetic tensor $F_{\mu\nu}$ on $\mathcal{M}$ so that
the triple $(\mathcal{M},g, F_{\mu\nu})$ still satisfies $(\ref{eq:einsteinmaxwell})$. 

\begin{figure}
\begin{center}
\begin{tikzpicture}
\definecolor{light-gray}{gray}{.9}
\fill[light-gray] (-2,2)--(2,-2)--(0,-4) -- (-4,0); 
\draw (0,0) -- (2,-2) node[sloped,above,midway]{$\mathcal{C}^+$}; 
\draw (2,-2) --  (0,-4) node[sloped,below,midway]{$\mathcal{C}^-$};
\draw (0,-4) -- (-2,-2) node[sloped,below,midway]{$\mathcal{H}^-$}; 
\draw (-2,-2) -- (-4,0) node[sloped,below,midway]{$\mathcal{H}^+_B$}; 
\draw (-2,-2) -- (0,0) node[sloped,above,midway]{$\mathcal{H}^+_A$}; 
\draw (-4,0) -- (-2,2) node[sloped,above,midway]{$\mathcal{CH}^+_B$}; 
\draw (-2,2) -- (0,0) node[sloped,above,midway]{$\mathcal{CH}^+_A$}; 
\path [draw=black,fill=white] (0,0) circle (1/16); 
\path [draw=black,fill=black] (2,-2) circle (1/16) node[right]{$\mathcal{B}_c$}; 
\path [draw=black,fill=white] (0,-4) circle (1/16); 
\path [draw=black,fill=black] (-2,-2) circle (1/16) node[left]{$\mathcal{B}_+$}; 
\path [draw=black,fill=white] (-4,0) circle (1/16); 
\path [draw=black,fill=black] (-2,2) circle (1/16) node[left]{$\mathcal{B}_-$}; 

\draw (-2,0) node{$\mathcal{M}_{\rm interior}$}; 
\draw (0,-2) node{$\mathcal{M}_{\rm static}$}; 

\end{tikzpicture}
\end{center}
\caption{The manfiold-with-stratified-boundary $\mathcal{M}$}\label{fig2}
\end{figure}
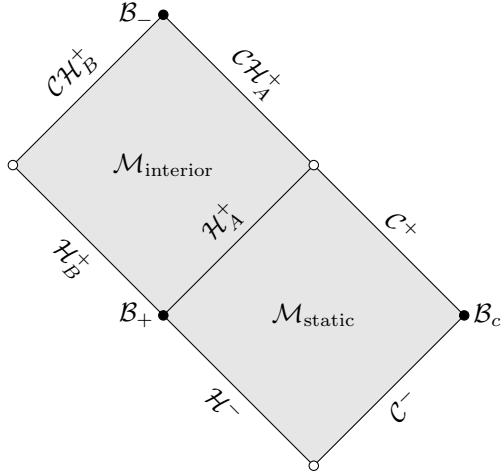

\section{Precise statement of the main theorem}\label{resultsection}

In this section, we will give a precise statement of the main result of this paper as Theorem~\ref{mainresult}
(see Section~\ref{mainresultsubse}).
This formulation only refers to the region $\mathcal{M}$ constructed explicitly
in the previous section. We shall then explain (see Section~\ref{redux})
how Theorem~\ref{maintheoremINTRO} of Section~\ref{theintrosec} can immediately
be reduced to this statement.

\subsection{Statement of Theorem~\ref{mainresult}}
\label{mainresultsubse}

Let $\mathcal{M}$ be defined by~\eqref{thisisM}.
We will consider below hypersurfaces-with-boundary $\Sigma \subset \mathcal{M}\setminus \mathcal{H}^+_B$ 
which are
connected, spacelike and compact, 
 transversally intersect $\mathcal{C}^+$ and $\mathcal{H}^+_A$ and have
boundary consisting of a single sphere in $\mathcal{M}_{{\rm interior}}$ and a single sphere in $\mathcal{C}^+$. (We shall soon specialise to the case where $\Sigma$ is itself
spherically symmetric.)

A solution $\psi$ to the wave equation is uniquely determined in $D^+_{\mathcal{M}}\left(\Sigma\right)$, the domain of dependence of $\Sigma$, by its corresponding Cauchy data along $\Sigma$. In particular, we have the following well-known proposition.

\begin{proposition}\label{local}
Given Cauchy data $\left(\Psi,\Psi'\right) \in H^s(\Sigma) \times H^{s-1}(\Sigma)$, there exists a weak solution $\psi$ to the wave equation
$(\ref{linearwaveequation})$
 in $D^+_{\mathcal{M}}(\Sigma)$ uniquely defined by the property that $\left(\psi,n_{\mathcal{S}}\psi\right) \subset
H^s_{{\rm loc}}(\mathcal{S})\times H^{s-1}_{{\rm loc}}(\mathcal{S})$ for any 
spacelike hypersurface $\mathcal{S} \in D^+_{\mathcal{M}}(\Sigma)$ and that $(\psi|_{\Sigma},n_{\Sigma}\psi|_{\Sigma}) = (\Psi,\Psi')$. Furthermore, if $\Sigma$ is spherically symmetric
and the data $(\Psi,\Psi')$ are spherically symmetric, then the solution
$\psi$ is also spherically symmetric.
 \end{proposition}

Our main result is the following theorem.
\begin{theorem}\label{mainresult}
Fix a non-degenerate tuple $(M,e,\Lambda)$ and consider the corresponding 
Reissner--Nordstr\"om--de Sitter metric $(\mathcal{M},g)$ and let $\Sigma$ denote
a spherically symmetric hypersurface as above. 
Then there exists $\epsilon = \epsilon(M,e,\Lambda) > 0$ and spherically symmetric Cauchy data $(\Psi,\Psi') \in H^{1+\epsilon}(\Sigma) \times H^{\epsilon}(\Sigma)$  such that the corresponding solution $\psi$ to the wave equation $(\ref{linearwaveequation})$ obtaining the Cauchy data 
satisfies
\begin{equation}
\label{toblowup123}
\|\psi\|_{\dot H^1(\mathcal{N})}=\infty,
\end{equation}
i.e.~$\psi\not\in H^1(\mathcal{N})$, where $\mathcal{N}$ is any constant $U_-$ hypersurface emanating from a sphere
in $D^+(\Sigma)\cap\mathcal{M}_{\rm interior}$ and terminating on  a sphere of $\mathcal{CH}^+_A$.
\end{theorem}
In Figure~\ref{fig3} we have depicted the hypersurface $\Sigma$, its future domain of 
dependence $D^+(\Sigma)$ and a choice of hypersurface $\mathcal{N}$.
The square of the homogeneous Sobolev norm $\dot H^1(\mathcal{N})$ can be interpreted as the energy flux measured
by a family of local observers on $\mathcal{N}$. Note that the statement $(\ref{toblowup123})$
is independent of the choice of induced volume form on $\mathcal{N}$.

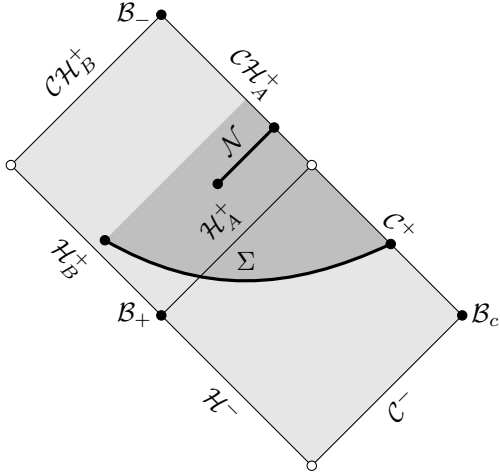
\begin{figure}
\begin{center}
\begin{tikzpicture}
\definecolor{light-gray}{gray}{.9}

\fill[light-gray] (-2,2)--(2,-2)--(0,-4) -- (-4,0); 
\fill[lightgray] (1.05,-1.05) to [out=205,in=-30] (-2.75,-1) -- (-.875,.875)--(1.05,-1.05);
\draw (0,0) -- (2,-2) node[sloped,above,midway]{$\mathcal{C}^+$}; 
\draw (2,-2) --  (0,-4) node[sloped,below,midway]{$\mathcal{C}^-$};
\draw (0,-4) -- (-2,-2) node[sloped,below,midway]{$\mathcal{H}^-$}; 
\draw (-2,-2) -- (-4,0) node[sloped,below,midway]{$\mathcal{H}^+_B$}; 
\draw (-2,-2) -- (0,0) node[sloped,above,midway]{$\mathcal{H}^+_A$}; 
\draw (-4,0) -- (-2,2) node[sloped,above,midway]{$\mathcal{CH}^+_B$}; 
\draw (-2,2) -- (0,0) node[sloped,above,midway]{$\mathcal{CH}^+_A$}; 
\path [draw=black,fill=white] (0,0) circle (1/16); 
\path [draw=black,fill=black] (2,-2) circle (1/16) node[right]{$\mathcal{B}_c$}; 
\path [draw=black,fill=white] (0,-4) circle (1/16); 
\path [draw=black,fill=black] (-2,-2) circle (1/16) node[left]{$\mathcal{B}_+$}; 
\path [draw=black,fill=white] (-4,0) circle (1/16); 
\path [draw=black,fill=black] (-2,2) circle (1/16) node[left]{$\mathcal{B}_-$}; 

\draw[very thick] (1.05,-1.05) to [out=205,in=-30]node[sloped,above,midway]{$\Sigma$} (-2.75,-1) ;
\path [draw = black, fill = black] (1.05,-1.05) circle (1/16);
\path [draw = black, fill = black] (-2.75,-1) circle (1/16);

\draw[very thick] (-1.25,-.25) -- (-.5,.5) node[sloped,above,midway]{$\mathcal{N}$};
\path [draw = black, fill = black] (-.5,.5) circle (1/16);
\path [draw = black, fill = black] (-1.25,-.25) circle (1/16);
\end{tikzpicture}
\end{center}

\caption{The manifold $\mathcal{M}$ with hypersurfaces $\Sigma$, $\mathcal{N}$ and shaded domain of influence of $\Sigma$}\label{fig3}
\end{figure}

We will provide two proofs of this result. The first proof, given in Section~\ref{modes}, will be a direct construction based on individual mode solutions. The second proof, given in Section~\ref{timetranslate} will be based on the time translation invariance of the spacetime and is an adaption to the cosmological setting of our arguments from~\cite{DafShlap}. 

We note that, it will be immediate from either proof that instead of spherically symmetric
$(\Psi,\Psi')$ in Theorem~\ref{mainresult},  we can take   data 
 supported instead on an arbitrary fixed higher
spherical harmonic number $\ell\ge 1$. 

\subsection{Reduction of Theorem~\ref{maintheoremINTRO} to Theorem~\ref{mainresult}}
\label{redux}

We now explain how Theorem~\ref{maintheoremINTRO} can be inferred from Theorem~\ref{mainresult}.

First of all, we observe that by
local existence considerations similar to Proposition~\ref{local},
it clearly suffices to establish Theorem~\ref{maintheoremINTRO} for a spherically symmetric $\widetilde\Sigma$. In the rest of this section we shall thus work with spherically symmetric hypersurfaces. Next, we note that $\left(\mathcal{M}',g\right)$ is isometric to $\left(\mathcal{M},g\right)$, where we define $\mathcal{M}' \doteq \mathcal{M}'_{{\rm static}}\cup\mathcal{H}^+_B\cup\mathcal{M}_{{\rm interior}}$ (see Figure~\ref{fig4}). Hence we obtain the analogue of Theorem~\ref{mainresult} for the region $\mathcal{M}'$. 

Next, 
using the finite speed of propagation, we immediately are able to obtain a solution $\psi$ arising from Cauchy data $\left(\Psi,\Psi'\right) \in H^{1+\epsilon}(\widetilde\Sigma) \times H^{\epsilon}(\widetilde\Sigma)$ such that $\left\vert\left\vert \psi\right\vert\right\vert_{H^1\left(\mathcal{N}\right)} =\left\vert\left\vert \psi\right\vert\right\vert_{H^1\left(\mathcal{\underline{N}}\right)} = \infty$, where $\mathcal{\underline{N}}$ is a null hypersurface transversally intersecting $\mathcal{CH}^+_B$ and $\mathcal{N}$ is a null hypersurface transversally intersecting $\mathcal{CH}^+_A$. 

Finally, it remains to show that the solution has an infinite $H^1$ norm along hypersurfaces transversally intersecting $\mathcal{CH}^+$ which lie outside the shaded region. This follows from a straightforward propagation of singularities argument, using however also
the local finiteness of the energy flux of $\psi$ \emph{along} 
$\mathcal{CH}^+_{A,B}$,
a result proven in~\cite{annefranzen}.

\begin{figure}
\begin{center}
\begin{tikzpicture}
\definecolor{light-gray}{gray}{.9}

\fill[light-gray] (-2,2)--(2,-2)--(0,-4) -- (-2,-2) -- (-4,-4) -- (-6,-2)-- (-2,2); 
\fill[lightgray] (.5,-.5)  to [out = 200,in = 10] (-1.75,-1.1) -- (-.325,.325) -- (.5,-.5);
\fill[lightgray] (-2.25,-1.1) to [out = 170, in = -20] (-4.5,-.5) -- (-3.675,.325) --(-2.25,-1.1);
\draw (0,0) -- (2,-2) node[sloped,above,midway]{$\mathcal{C}^+$}; 
\draw (2,-2) --  (0,-4) node[sloped,below,midway]{$\mathcal{C}^-$};
\draw (0,-4) -- (-2,-2);
\draw (-2,-2) -- (-4,0) node[sloped,below,midway]{$\mathcal{H}^+_B$}; 
\draw (-2,-2) -- (0,0) node[sloped,below,midway]{$\mathcal{H}^+_A$}; 
\draw (-4,0) -- (-2,2) node[sloped,above,midway]{$\mathcal{CH}^+_B$}; 
\draw (-2,2) -- (0,0) node[sloped,above,midway]{$\mathcal{CH}^+_A$}; 
\draw (-4,0) -- (-6,-2) node[sloped,above,midway]{$\mathcal{C}^+$};
\draw (-6,-2) -- (-4,-4) node[sloped,below,midway]{$\mathcal{C}^-$};
\draw (-4,-4) -- (-2,-2);

\path [draw=black,fill=white] (0,0) circle (1/16); 
\path [draw=black,fill=black] (2,-2) circle (1/16) node[right]{$\mathcal{B}_c$}; 
\path [draw=black,fill=white] (0,-4) circle (1/16); 
\path [draw = black,fill = white] (-4,-4) circle (1/16);
\path [draw = black, fill = black] (-6,-2) circle (1/16);
\path [draw=black,fill=black] (-2,-2) circle (1/16) node[left]{$\mathcal{B}_+$}; 
\path [draw=black,fill=white] (-4,0) circle (1/16); 
\path [draw=black,fill=black] (-2,2) circle (1/16) node[left]{$\mathcal{B}_-$}; 

\draw[very thick] (-.9,-.6) -- (-.15,.15) node[sloped,above,midway]{$\mathcal{N}$};
\draw[very thick] (-3.85,.15) -- (-3.1,-.6) node[sloped,above,midway]{$\mathcal{\underline{N}}$};
\draw[very thick] (.5,-.5) to [out = 200,in = 10] (-1.75,-1.1) to [out = 190,in = -10]   (-2.25,-1.1) to [out = 170, in = -20] (-4.5,-.5);

\draw (0,-2) node{$\mathcal{M}_{{\rm static}}$};
\draw (-4,-2) node{$\mathcal{M}'_{{\rm static}}$};
\draw (-2,0) node{$\mathcal{M}_{{\rm interior}}$};
\draw (-2,-1.35) node{$\widetilde{\Sigma}$};
\path [draw=black,fill=black] (.5,-.5) circle (1/16);
\path [draw=black,fill=black] (-4.5,-.5) circle (1/16);
\path [draw=black,fill=black] (-.15,.15) circle (1/16);
\path [draw=black,fill=black] (-3.85,.15) circle (1/16);
\path [draw=black,fill=black] (-3.1,-.6) circle (1/16);
\path [draw=black,fill=black] (-.9,-.6) circle (1/16);
\end{tikzpicture}
\end{center}
\caption{The manifold $\mathcal{M}'\cup\mathcal{M}$ with the hypersurfaces $\widetilde\Sigma$, $\mathcal{\underline{N}}$ and $\mathcal{N}$}\label{fig4}
\end{figure}
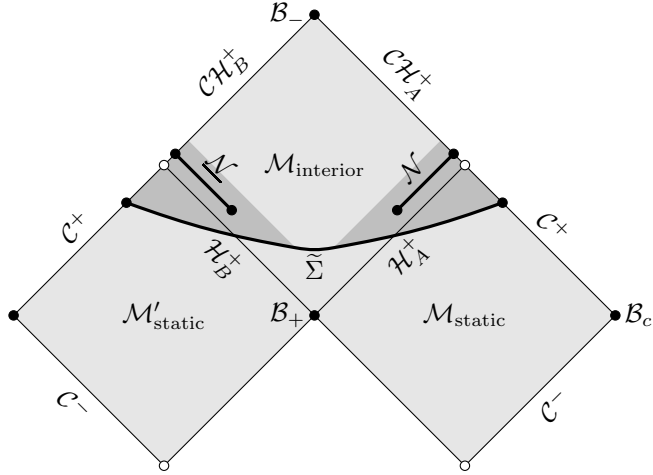

\section{First proof of Theorem~\ref{mainresult}: an explicit mode-solution construction}\label{modes}
In this section, we will give our first proof of Theorem~\ref{mainresult} by constructing the desired solution $\psi$ via direct o.d.e.~analysis. 

We start in Section~\ref{sepvar} by reviewing the separation of variables for the wave equation~\eqref{linearwaveequation}. In Section~\ref{asympanal} we carry out an asymptotic analysis of the radial o.d.e.~at its singular points. Next, in Section~\ref{nontrivial} we establish two o.d.e.~lemmata which have the interpretation as statements of non-triviality of reflection and transmission. These lemmata will also be used in Section~\ref{timetranslate}. Finally, in Section~\ref{theproofode} we give the proof of Theorem~\ref{mainresult}.

\subsection{Separation of variables on the Reissner--Nordstr\"{o}m--de Sitter spacetime}\label{sepvar}
In this section we will review the procedure for separation of variables for the wave equation
$(\ref{linearwaveequation})$ on the Reissner--Nordstr\"{o}m--de Sitter spacetime.

We say that a spherically symmetric $\psi : \mathcal{M}_{{\rm static}} \to \mathbb{C}$ 
in the static region is a \emph{mode solution} if $\psi$ satisfies~\eqref{waveeqn} and there exists $\omega \in \mathbb{C}$ so that
\begin{equation}\label{amode}
\psi(t,r) = e^{-it\omega}R(r).
\end{equation}
Note that a spherically symmetric function $\psi(t,r)$ satisfies the wave equation $(\ref{linearwaveequation})$ in $\mathcal{M}_{{\rm static}}$ if and only if
\begin{equation}\label{waveeqn}
\partial_t^2\psi - r^{-2}\mu(r)\partial_r\left(r^2\mu(r)\partial_r\psi\right) = 0,
\end{equation}
where we have defined
\[\mu(r) \doteq 1-\frac{2M}{r} + \frac{e^2}{r^2} - \frac{\Lambda}{3}r^2.\]
In particular, in the case of a mode solution, the function $R(r)$ from~\eqref{amode} will satisfy\footnote{For the analysis of solutions supported
on a fixed higher spherical harmonic number $\ell$, 
one must simply replace use of $(\ref{radode})$ everywhere in what follows  by the resulting
o.d.e.~depending on $\ell$.} the \emph{radial o.d.e.}:
\begin{equation}\label{radode}
r^{-2}\mu(r)\frac{d}{dr}\left(r^2\mu(r)\frac{dR}{dr}\right) + \omega^2R = 0,\qquad r \in  (r_+,r_c).
\end{equation}
Finally, we also obtain a definition, \emph{mutatis mutandis}, for spherically symmetric mode solutions $\psi : \mathcal{M}_{{\rm interior}} \to \mathbb{C}$ in the interior region, 
where the $r$-range of $(\ref{radode})$
is replaced by $r\in (r_-,r_+)$. 

\subsection{Asymptotic analysis of the radial o.d.e., basis solutions and analyticity
properties}\label{asympanal}
The o.d.e.~\eqref{radode} has regular singular points at $r = r_c$, $r = r_+$ and $r = r_-$. In particular, if $R$ satisfies~\eqref{radode}, then we will have 
\begin{equation}\label{nearrc}
\frac{d^2R}{dr^2} + \left(\frac{1}{r-r_c} + O\left(1\right)\right)\frac{dR}{dr} + \left(\frac{\omega^2\kappa_c^{-2}}{4(r-r_c)^2} + O\left((r-r_c)^{-1}\right)\right)R = 0,\text{ for }|r-r_c| \ll 1,
\end{equation}
\begin{equation}\label{nearrplus}
\frac{d^2R}{dr^2} + \left(\frac{1}{r-r_+} + O\left(1\right)\right)\frac{dR}{dr} + \left(\frac{\omega^2\kappa_+^{-2}}{4(r-r_+)^2} + O\left((r-r_+)^{-1}\right)\right)R = 0,\text{ for }|r-r_+| \ll 1,
\end{equation}
\begin{equation}\label{nearrminus}
\frac{d^2R}{dr^2} + \left(\frac{1}{r-r_-} + O\left(1\right)\right)\frac{dR}{dr} + \left(\frac{\omega^2\kappa_-^{-2}}{4(r-r_-)^2} + O\left((r-r_-)^{-1}\right)\right)R = 0,\text{ for }|r-r_-| \ll 1.
\end{equation}

Then the theory of o.d.e.'s with regular singular points (see~\cite{olver} and also Appendix A of~\cite{yakovinsta}) immediately yields the following lemma.

\begin{lemma}\label{odesworkwell}Let $\mathcal{U} = \{0\} \cup \{\omega \in \mathbb{C}: \text{Im}\left(\omega\right) \not\in \kappa_c\mathbb{Z} \cup \kappa_+\mathbb{Z}\cup \kappa_-\mathbb{Z} \}$. Then for every $\omega \in \mathcal{U}$ there exist six solutions $R_{{\rm in},+}\left(\omega,r\right)$, $R_{{\rm out},+}\left(\omega,r\right)$, $R_{{\rm in},c}\left(\omega,r\right)$, $R_{{\rm out},c}\left(\omega,r\right)$, $R_{{\rm in},-}\left(\omega,r\right)$ and $R_{{\rm out},-}\left(\omega,r\right)$ to~\eqref{radode} defined by
\begin{enumerate}
	\item $R_{{\rm in},+}(\omega,r) = \left(r-r_+\right)^{-\frac{i\omega}{2\kappa_+}}\hat{R}_{{\rm in},+}\left(\omega,r\right), \text{ where }\hat{R}_{{\rm in},+}\left(\omega,r\right)$ is analytic for $r \in (r_-,r_c)$  and satisfies $\hat{R}_{{\rm in},+}\left(\omega,r_+\right) = 1$.
	\item $R_{{\rm out},+}(\omega,r) = \left(r-r_+\right)^{\frac{i\omega}{2\kappa_+}}\hat{R}_{{\rm out},+}\left(\omega,r\right), \text{ where }\hat{R}_{{\rm out},+}\left(\omega,r\right)$is analytic for $r \in (r_-,r_c)$  and satisfies $\hat{R}_{{\rm out},+}\left(\omega,r_+\right) = 1$.

	\item $R_{{\rm in},c}(\omega,r) = \left(r-r_c\right)^{-\frac{i\omega}{2\kappa_c}}\hat{R}_{{\rm in},c}\left(\omega,r\right), \text{ where }\hat{R}_{{\rm in},c}\left(\omega,r\right)$ is analytic for $r \in (r_-,r_+) \cup (r_+,r_c]$  and satisfies $\hat{R}_{{\rm in},c}\left(\omega,r_c\right) = 1$.

	\item $R_{{\rm out},c}(\omega,r) = \left(r-r_c\right)^{\frac{i\omega}{2\kappa_c}}\hat{R}_{{\rm out},c}\left(\omega,r\right), \text{ where }\hat{R}_{{\rm out},c}\left(\omega,r\right)$ is analytic for  $r \in (r_-,r_+) \cup (r_+,r_c]$  and satisfies $\hat{R}_{{\rm out},c}\left(\omega,r_c\right) = 1$.

	\item $R_{{\rm in},-}(\omega,r) = \left(r-r_-\right)^{-\frac{i\omega}{2\kappa_-}}\hat{R}_{{\rm in},-}\left(\omega,r\right), \text{ where }\hat{R}_{{\rm in},-}\left(\omega,r\right)$ is analytic for $r \in [r_-,r_+) \cup (r_+,r_c)$  and satisfies $\hat{R}_{{\rm in},-}\left(\omega,r_-\right) = 1$.

	\item $R_{{\rm out},-}(\omega,r) = \left(r-r_-\right)^{\frac{i\omega}{2\kappa_-}}\hat{R}_{{\rm out},-}\left(\omega,r\right), \text{ where }\hat{R}_{{\rm out},-}\left(\omega,r\right)$ is analytic for  $r \in [r_-,r_+) \cup (r_+,r_c)$ and satisfies $\hat{R}_{{\rm out},-}\left(\omega,r_-\right) = 1$.

\end{enumerate}

Furthermore, considered as functions of $\omega$, $\hat{R}_{{\rm in},+}\left(\omega,r\right)$, $\hat{R}_{{\rm out},+}\left(\omega,r\right)$, $\hat{R}_{{\rm in},c}\left(\omega,r\right)$, $\hat{R}_{{\rm out},c}\left(\omega,r\right)$,  $\hat{R}_{{\rm in},-}\left(\omega,r\right)$ and $\hat{R}_{{\rm out},-}\left(\omega,r\right)$ are holomorphic for $\omega \in \mathcal{U}$.

Finally, it is immediate that when $\omega \in \mathcal{U}\setminus\{0\}$, $R_{{\rm in},+}(\omega,r)$ and $R_{{\rm out},+}(\omega,r)$ are linearly independent. Similarly, $R_{{\rm in},c}(\omega,r)$ and $R_{{\rm out},c}(\omega,r)$ are linearly independent and $R_{{\rm in},-}(\omega,r)$ and $R_{{\rm out},-}(\omega,r)$ are linearly independent. 

\end{lemma}
\begin{proof}The expressions~\eqref{nearrc},~\eqref{nearrplus} and~\eqref{nearrminus} imply that the indicial roots near $r_c$, $r_+$ and $r_-$ are $\pm\frac{i\omega}{2\kappa_c}$, $\pm\frac{i\omega}{2\kappa_+}$ and $\pm\frac{i\omega}{2\kappa_-}$respectively. The desired solutions may be then constructed by explicit power series (see~\cite{olver} and also Appendix A of~\cite{yakovinsta}) as long we do not allow $\omega$ to enter a region where the difference of two indicial roots is a non-zero integer. We immediately see that the constructed solutions are valid for $\omega \in \mathcal{U}$.
\end{proof}
\begin{remark}When $\omega = 0$, the final statement of above lemma does not apply. This fact plays no role in the proof of Theorem~\ref{mainresult}; nevertheless, we note that one could construct additional linearly independent solutions which must, however, contain logarithmic singularities. 

Though it is also not relevant for the proof of Theorem~\ref{mainresult}; we note that at $\omega = i\kappa_+$, the solution $R_{{\rm out},+}(\omega,r)$ has a simple pole, and that at $\omega = -i\kappa_+$, the solution $R_{{\rm in},+}\left(\omega,r\right)$ has a simple pole (cf.~the analysis of the reflexion and transmission coefficients considered in~\cite{CH}). Analogous statements hold for $R_{{\rm out},c}(\omega,r)$, $R_{{\rm in},c}(\omega,r)$,  $R_{{\rm out},-}(\omega,r)$ and $R_{{\rm in},-}(\omega,r)$.
\end{remark}

The following lemma will be useful.
\begin{lemma}\label{analyticcoefficients}There exist holomorphic functions $\mathfrak{A} : \mathcal{U}\setminus\{0\} \to \mathbb{C}$, $\mathfrak B: \mathcal{U}\setminus\{0\} \to \mathbb{C}$, $\widetilde{\mathfrak{A}}:\mathcal{U}\setminus\{0\} \to \mathbb{C}$ and $\mathfrak{\widetilde B}:\mathcal{U}\setminus\{0\}\to \mathbb{C}$ so that for all $r \not\in \{r_-,r_+,r_c\}$ and $\omega \in \mathcal{U}\setminus\{0\}$ we have
\begin{equation}\label{firstexpansion0}
R_{{\rm in},c}\left(\omega,r\right) = \mathfrak{A}(\omega)R_{{\rm in},+}\left(\omega,r\right) + \mathfrak{B}(\omega)R_{{\rm out},+}\left(\omega,r\right),
\end{equation}
\begin{equation}\label{secondexpansion0}
R_{{\rm in},+}\left(\omega,r\right) = \mathfrak{\widetilde A}(\omega) R_{{\rm in},-}(\omega,r) +\mathfrak{ \widetilde B}(\omega) R_{{\rm out},-}(\omega,r).
\end{equation}
\end{lemma}
\begin{proof}The fact that the expressions~\eqref{firstexpansion0} and~\eqref{secondexpansion0} hold for functions $\mathfrak A$, $\mathfrak B$, $\mathfrak{\widetilde A}$ and $\mathfrak{\widetilde B}$ of unspecified regularity is an immediate consequence of the statement from Lemma~\ref{odesworkwell} that $R_{{\rm in},+}\left(\omega,r\right)$ and $R_{{\rm out},+}\left(\omega,r\right)$ are linearly independent and  $R_{{\rm in},-}(\omega,r)$ and  $R_{{\rm out},-}(\omega,r)$ are linearly independent when $\omega \in \mathcal{U}\setminus\{0\}$.

In order to analyse the regularity of $\mathfrak{A}$, $\mathfrak B$, $\mathfrak{\widetilde A}$ and $\mathfrak{\widetilde B}$, we use the Wronskian $\mathfrak W$. Given two functions $f(r)$ and $g(r)$ we define the Wronskian:
\[\mathfrak W\left(f,g\right)(r) \doteq r^2\mu(r)\frac{df}{dr}g - r^2\mu(r)\frac{dg}{dr}f.\]
A key fact is that if $f$ and $g$ satisfy the radial o.d.e., then
\begin{equation}\label{wronkcons}
\frac{d\mathfrak W}{dr} = 0.
\end{equation}

Let $r_0 = \frac{r_++r_-}{2}$. Using~\eqref{wronkcons} and the observation that $\mathfrak W\left(f,f\right) = 0$, one may establish the formula
\begin{equation}\label{formA}
\mathfrak A\left(\omega\right) = \frac{\mathfrak W\left(R_{{\rm in},c}\left(\omega,r_0\right),R_{{\rm out},+}\left(\omega,r_0\right)\right)}{\mathfrak W\left(R_{{\rm in},+}\left(\omega,r_0\right),R_{{\rm out},+}\left(\omega,r_0\right)\right)}.
\end{equation}
The asymptotic behavior of $R_{{\rm in},+}$ and $R_{{\rm out},+}$ and~\eqref{wronkcons} imply that $\mathfrak W\left(R_{{\rm in},+}\left(\omega,r_0\right),R_{{\rm out},+}\left(\omega,r_0\right)\right)\ne0$. Thus, \eqref{formA} and Lemma~\ref{odesworkwell} yield that $\mathfrak A$ is 
holomorphic. The arguments for $\mathfrak B$, $\mathfrak{\widetilde A}$ and $\mathfrak{\widetilde B}$ are analogous. 
\end{proof}

\subsection{Non-triviality of reflexion and transmission}\label{nontrivial}

Before giving the proof of Theorem~\ref{mainresult}, we need two final preparatory lemmata which can be interpreted as the statement that neither the reflexion map in the static region nor the transmission map in the black hole interior
region can vanish identically (cf.~the proof of Theorem 10 
of~\cite{Dafermos:2014jwa} and the results of~\cite{kehle}).

\begin{lemma}\label{unique1}There exists $\omega \in \mathbb{R}$ so that the solutions $R_{{\rm out},+}\left(\omega,r\right)$ and $R_{{\rm in},c}\left(\omega,r\right)$ are linearly independent. 

\end{lemma}
\begin{proof}Suppose, for the sake of contradiction, that for every $\omega \in \mathbb{R}$ there is some $a(\omega) \in \mathbb{C}$ so that $R_{{\rm out},+}\left(\omega,r\right) = a(\omega)R_{{\rm in},c}\left(\omega,r\right)$. 

Set $R\left(r\right) \doteq R_{{\rm out},+}\left(\omega,r\right)$. We have that $R$ will satisfy the radial o.d.e.~\eqref{radode}. It is convenient to define a new $\tilde r(r)$ coordinate by
\[\frac{d\tilde r}{dr} = \left(r^2\mu\right)^{-1},\qquad \tilde r\left(\frac{r_++r_c}{2}\right) = 0.\]
Then~\eqref{radode} becomes
\begin{equation}\label{newradode}
\frac{d^2R}{d{\tilde r}^2} + r^4\omega^2 R = 0.
\end{equation}

Now define
\[Q_T \doteq \text{Im}\left(\frac{dR}{d\tilde r}\overline{R}\right).\]
We clearly have
\[\frac{dQ_T}{d\tilde r} = 0.\]
In particular, using the fundamental theorem of calculus, we obtain
\begin{align}\label{energyid}
0 = \lim_{r\to r_c}Q_T\left(r\right) - \lim_{r\to r_+}Q_T\left(r\right).
\end{align}
Let us now compute $\lim_{r\to r_c}Q_T\left(r\right)$:
\begin{align}\label{atrc}
\lim_{r\to r_c}\text{Im}\left(\frac{dR}{d\tilde r}\overline{R}\right) &= \lim_{r\to r_c} \text{Im}\left(r^2\mu(r)\frac{dR}{dr}\overline{R}\right)
\\ \nonumber &= -2r_c^2\kappa_c\left|a\right|^2\lim_{r\to r_c}\text{Im}\left((r-r_c)\frac{d}{dr}\left(\left(r-r_c\right)^{\frac{-i\omega}{2\kappa_c}}\right)\left(r-r_c\right)^{\frac{i\omega}{2\kappa_c}}\right)
\\ \nonumber &= r_c^2\omega\left|a\right|^2.
\end{align}
Similarly, we may compute
\begin{equation}\label{atrplus}
\lim_{r\to r_+}\text{Im}\left(\frac{dR}{d\tilde r}\overline{R}\right)  = r_+^2\omega.
\end{equation}
Combining~\eqref{energyid},~\eqref{atrc} and~\eqref{atrplus} yields, when $\omega \neq 0$,
\begin{equation}\label{whatwehavetolivewith}
\left|a\right|^2 = \frac{r_+^2}{r_c^2} < 1.
\end{equation}
However, since $a$ must be smooth in $\omega$, by continuity we also have~\eqref{whatwehavetolivewith} when $\omega = 0$.

On the other hand, when $\omega = 0$, the o.d.e.~\eqref{newradode} can be easily solved explicitly and we see that
\[R\left(0,r\right) = 1 \Rightarrow a(0) = 1.\]
This contradicts~\eqref{whatwehavetolivewith}, completing the proof.

\end{proof}

\begin{lemma}\label{unique2}For every $\omega \in \mathbb{R}\setminus\{0\}$ the solutions $R_{{\rm in},+}\left(\omega,r\right)$ and $R_{{\rm in},-}\left(\omega,r\right)$ are linearly independent. 
\end{lemma}
\begin{proof}We proceed in a similar fashion to the proof of Lemma~\ref{unique1}. Fix an arbitrary $\omega \in \mathbb{R}\setminus\{0\}$ and suppose, for the sake of contradiction, that
\[R_{{\rm in},+}(\omega,r) = a R_{{\rm in},-}\left(\omega,r\right)\]
for some $a \in \mathbb{C}$. 

We set $R(r) = R_{{\rm in},+}(\omega,r)$ and introduce the $\tilde r(r)$ coordinate by
\[\frac{d\tilde r}{dr} = \left(r^2\mu\right)^{-1},\qquad \tilde r\left(\frac{r_-+r_+}{2}\right) = 0,\]
Then~\eqref{radode} becomes
\begin{equation}\label{newradode}
\frac{d^2R}{d{\tilde r}^2} + r^4\omega^2 R = 0,
\end{equation}
and, just as in the proof of Lemma~\ref{unique1}, we see that if we set
\[Q_T \doteq \text{Im}\left(\frac{dR}{d\tilde r}\overline{R}\right),\]
then
\begin{equation}\label{ohmygoditisconserved}
\frac{d}{d\tilde r}Q_T = 0.
\end{equation}

One computes
\[Q_T\left(r_+\right) = -r_+^2\omega,\qquad Q_T\left(r_-\right) = |a|^2r_-^2\omega.\]
Together with~\eqref{ohmygoditisconserved} we obtain
\[|a|^2r_-^2 + r_+^2 = 0,\]
which is clearly a contradiction.
\end{proof}

\subsection{Proof of Theorem~\ref{mainresult}}\label{theproofode}

Now we are ready to prove Theorem~\ref{mainresult}.
\begin{proof}
Recall the fundamental inequality $\kappa_+<\kappa_-$ satisfied by the surface gravities
(Lemma~\ref{fundsurfgrav}).
We may thus choose $\hat{\kappa}$ satisfying $\kappa_+ < \hat{\kappa} < \kappa_-$. 
We also introduce a parameter $\omega_R \in \mathbb{R}\setminus \{0\}$ which will be fixed later in the proof. 

We start by defining a mode solution $\psi$ in the static region
$\mathcal{M}_{\rm static}$ where $r \in (r_+,r_c)$ by
\begin{align}\label{thisisus}
\psi\left(t,r\right) &\doteq e^{-it\left(\omega_R - i\frac{\hat{\kappa}}{2}\right)}R_{{\rm in},c}\left(\omega_R - i\frac{\hat{\kappa}}{2},r\right)
\\ \nonumber &= e^{-it\left(\omega_R - i\frac{\hat{\kappa}}{2}\right)}\left(r-r_c\right)^{-\frac{i\left(\omega_R - i\frac{\hat{\kappa}}{2}\right)}{2\kappa_c}}\hat{R}_{{\rm in},c}\left(\omega_R - i\frac{\hat{\kappa}}{2},r\right).
\end{align}

It follows from~\eqref{rstar} that near $r = r_c$, we have
\[\psi(t,r) = e^{-i\left(\omega_R-i\frac{\hat{\kappa}}{2}\right)\left(t-r^*\right)}\tilde{R}\left(r\right),\]
where $\tilde{R}(r)$ is analytic for $r$ near $r_c$. In particular, it is clear by working in the ingoing Eddington--Finkelstein coordinates $\left(u,r,\theta,\phi\right)$ of Section~\ref{ingoingfir} that $\psi$ extends smoothly to $\mathcal{C}^+$. 

Next, we want to extend $\psi$ to $\mathcal{M}_{{\rm interior}}$ also as a mode solution. First, using Lemma~\ref{analyticcoefficients}, we expand
\begin{equation}\label{firstexpansion}
\psi(t,r) = e^{-it\left(\omega_R -i\frac{\hat{\kappa}}{2}\right)}\left[\mathfrak A\left(\omega_R-i\frac{\hat{\kappa}}{2}\right)R_{{\rm in},+}\left(\omega_R-i\frac{\hat{\kappa}}{2},r\right) + \mathfrak B\left(\omega_R-i\frac{\hat{\kappa}}{2}\right)R_{{\rm out},+}\left(\omega_R-i\frac{\hat{\kappa}}{2},r\right)\right].
\end{equation}
Then we define $\psi$ in the interior region where $r \in (r_-,r_+)$  by the following:\[\ \psi(t,r) = \mathfrak A\left(\omega_R-i\frac{\hat{\kappa}}{2}\right)e^{-it\left(\omega_R -i\frac{\hat{\kappa}}{2}\right)}R_{{\rm in},+}\left(\omega_R-i\frac{\hat{\kappa}}{2},r\right) .\]

We observe that in outgoing Eddington--Finklestein coordinates $\left(v,r,\theta,\phi\right)$
of Section~\ref{outgoingcorfir}, the expression~\eqref{firstexpansion} becomes
\begin{align}\label{secondexpansion}
&\psi(t,r) = e^{-iv\left(\omega_R - i\frac{\hat{\kappa}}{2}\right)}\Bigg[ \mathfrak A\left(\omega_R-i\frac{\hat{\kappa}}{2}\right)\hat{R}_{{\rm in},+}\left(\omega_R-i\frac{\hat{\kappa}}{2},r\right) 
\\ \nonumber &\qquad \qquad +\mathfrak B\left(\omega_R-i\frac{\hat{\kappa}}{2}\right)\left(r-r_+\right)^{\frac{i\left(\omega_R-i\frac{\hat{\kappa}}{2}\right)}{\kappa_+}}\hat{R}_{{\rm out},+}\left(\omega_R-i\frac{\hat{\kappa}}{2},r\right)\Bigg].
\end{align}
In particular, it is clear that if we now consider $\psi$ as a function on $\mathcal{M}_{{\rm static}}\cup\mathcal{H}^+_A \cup\mathcal{M}_{{\rm interior}}$, it continuously extends to the horizon $\mathcal{H}^+_A$. Next, we claim that irrespective of the values of $\mathfrak A\left(\omega_R-i\frac{\hat{\kappa}}{2}\right)$ and $\mathfrak B\left(\omega_R-i\frac{\hat{\kappa}}{2}\right)$, there exists a sufficiently small $\epsilon > 0$ so that we will have that 
\begin{equation}\label{regS}
\left(\psi|_{\mathcal{S}},n_{\mathcal{S}}\psi|_{\mathcal{S}}\right)\in H^{1+\epsilon}\left(\mathcal{S}\right) \times H^{\epsilon}\left(\mathcal{S}\right)
\end{equation}
for any compact spacelike hypersurface-with-boundary $\mathcal{S} \subset \mathcal{M}_{{\rm static}}\cup\mathcal{H}^+_A\cup\mathcal{M}_{{\rm interior}}$. Indeed, since the terms $e^{-iv\left(\omega_R - i\frac{\hat{\kappa}}{2}\right)}$, $\hat{R}_{{\rm in},+}\left(\omega_R-i\frac{\hat{\kappa}}{2},r\right)$ and $\hat{R}_{{\rm out},+}\left(\omega_R-i\frac{\hat{\kappa}}{2},r\right)$  are smooth expressions for $r \in (r_-,r_c)$ in the $(v,r,\theta,\phi)$ coordinates, it suffices to examine $\left(r-r_+\right)^{\frac{i\left(\omega_R-i\frac{\hat{\kappa}}{2}\right)}{\kappa_+}}$. However, since $\frac{\hat{\kappa}}{\kappa_+} > 1$, it is immediate that $\left(r-r_+\right)^{\frac{i\left(\omega_R-i\frac{\hat{\kappa}}{2}\right)}{\kappa_+}}$ lies in $H^{1+\epsilon}\left([r_+,r_++1]\right)$ for suitable $0 < \epsilon \ll 1$ and we thus obtain~\eqref{regS}.

Next we note that $\psi$ is easily seen to be a weak solution of the wave equation
$(\ref{linearwaveequation})$, and, 
furthermore, defining $(\Psi,\Psi')\doteq (\psi|_{\Sigma},n_{\Sigma}\psi|_{\Sigma})$, we have shown that
\[(\Psi,\Psi') \in H^{1+\epsilon}(\Sigma) \times H^{\epsilon}(\Sigma),\]
and thus $\psi$ coincides in $D^+(\Sigma)$
with the solution produced by Proposition~\ref{local}.

Now, to finish the proof it suffices to show that  $\psi$ satisfies 
$(\ref{toblowup123})$
for any spherically symmetric null hypersurface $\mathcal{N}$ intersecting $\mathcal{CH}^+_A$ transversally.  To see why this is true,  we use Lemma~\ref{analyticcoefficients} to expand $\psi$ near $r=r_-$ as
\begin{align*}
&\psi(t,r) = \mathfrak A\left(\omega_R-i\frac{\hat{\kappa}}{2}\right)e^{-it\left(\omega_R -i\frac{\hat{\kappa}}{2}\right)}\Bigg[\mathfrak{\widetilde{A}}\left(\omega_R-i\frac{\hat{\kappa}}{2}\right)R_{{\rm in},-}\left(\omega_R-i\frac{\hat{\kappa}}{2},r\right) \\ \nonumber &\qquad\qquad + \mathfrak{\widetilde{B}}\left(\omega_R-i\frac{\hat{\kappa}}{2}\right)R_{{\rm out},-}\left(\omega_R-i\frac{\hat{\kappa}}{2},r\right) \Bigg].
\end{align*}

In order to understand the regularity as we approach $\mathcal{CH}^+_A$, we re-write this in 
the ingoing coordinates $(\tilde u,r,\theta,\phi)$ of Section~\ref{ingoingcoorsec}. We obtain
\begin{align*}
&\psi(\tilde u,r) = \mathfrak A\left(\omega_R-i\frac{\hat{\kappa}}{2}\right)e^{i\tilde u\left(\omega_R -i\frac{\hat{\kappa}}{2}\right)}\Bigg[\mathfrak{\widetilde{A}}\left(\omega_R-i\frac{\hat{\kappa}}{2}\right)\hat{R}_{{\rm in},-}\left(\omega_R-i\frac{\hat{\kappa}}{2},r\right) 
\\ \nonumber &\qquad \qquad + \mathfrak{\widetilde{B}}\left(\omega_R-i\frac{\hat{\kappa}}{2}\right)\left(r-r_-\right)^{\frac{i\left(\omega_R - i\frac{\hat{\kappa}}{2}\right)}{\kappa_-}}\hat{R}_{{\rm out},-}\left(\omega_R-i\frac{\hat{\kappa}}{2},r\right) \Bigg].
\end{align*}

Since $\frac{\hat{\kappa}}{\kappa_-} < 1$, we have that 
\[\left\vert\left\vert \left(r-r_-\right)^{\frac{i\left(\omega_R - i\frac{\hat{\kappa}}{2}\right)}{\kappa_-}}\right\vert\right\vert_{H^1\left([r_-,r_-+1]\right)} = \infty.\]

In particular, if we can arrange so that $\mathfrak{A}\left(\omega_R-i\frac{\hat{\kappa}}{2}\right)\mathfrak{\widetilde B}\left(\omega_R-i\frac{\hat{\kappa}}{2}\right) \neq 0$, then we will have established~\eqref{toblowup123} and finished the proof. By Lemma~\ref{analyticcoefficients}, the functions $\mathfrak{A}\left(\omega\right)$ and $\mathfrak{\widetilde B}(\omega)$ are analytic. Thus, either (a) one of $\mathfrak{A}(\omega)$ or $\mathfrak{\widetilde B}(\omega)$ vanish identically, or 
(b) the zeros of $\mathfrak{A}(\omega)\mathfrak{\widetilde B}(\omega)$ form a discrete set. Fortunately, Lemmas~\ref{unique1} and~\ref{unique2} are sufficient to exclude possibility (a).  
Thus, we can indeed find an $\omega_R \in \mathbb{R}\setminus \{0\}$ such that $\mathfrak{A}\left(\omega_R-i\frac{\hat{\kappa}}{2}\right)\mathfrak{\widetilde B}\left(\omega_R-i\frac{\hat{\kappa}}{2}\right) \neq 0$ and the proof is concluded.

\end{proof}

\section{Second proof of Theorem~\ref{mainresult}: exploiting time-translation invariance}
\label{timetranslate}
In this section we will give our second proof of Theorem~\ref{mainresult} by exploiting time translation with respect to the Killing field $T$. We will in fact construct a solution to $(\ref{linearwaveequation})$
globally in the region $\mathcal{M}$ by prescribing 
scattering data on $(\mathcal{H}^+_B\cup\mathcal{B}_+\cup\mathcal{H}^-)\cup\mathcal{C}^-$
and then obtain Theorem~\ref{mainresult} by restricting the solution to $\mathcal{D}^+_{\mathcal{M}}\left(\Sigma\right)$.

We start in Section~\ref{prelimtime} by recalling various previous results which will be useful. Next, in Section~\ref{thekey} we prove Propositon~\ref{theargheart} which constructs a $1$-parameter family of data along $\mathcal{H}^-$, which are uniformly bounded $H^{1+\epsilon}(\mathcal{H}^-)$, so that the $H^1$ norm of the corresponding solution along $\mathcal{CH}^+_B$ can be arbitrarily large. Finally, in Section~\ref{finishhim} we use Proposition~\ref{theargheart} and the uniform boundedness principle to complete the proof of Theorem~\ref{mainresult}.

\subsection{Preliminary results}\label{prelimtime}
We start by introducing some notation and reviewing known results which will be used in the proof of Theorem~\ref{mainresult}. 

It is useful to introduce the following norms.
\begin{definition}For any spherically symmetric function $\psi$ such that $\frac{\partial\psi}{\partial v}|_{\mathcal{H}^-}$ is measurable, we set
\[\left\vert\left\vert \psi\right\vert\right\vert_{\dot{H}^1\left(\mathcal{H}^-\right)}^2 \doteq \int_{-\infty}^0\left|\frac{\partial \psi}{\partial u}|_{\mathcal{H}^-}\right|^2\, du + \int_{-1}^0\left|\frac{\partial \psi}{\partial U_+}\right|^2\, dU_+.\]

Similarly, if $\psi$ is spherically symmetric and $\frac{\partial\psi}{\partial v}|_{\mathcal{CH}^+_B}$ is measurable, we set
\[\left\vert\left\vert \psi\right\vert\right\vert_{\dot{H}^1\left(\mathcal{CH}^+_A\right)}^2 \doteq \int_{-\infty}^0\left|\frac{\partial \psi}{\partial v}|_{\mathcal{H}^-}\right|^2\, dv + \int_{-1}^0\left|\frac{\partial \psi}{\partial V_-}\right|^2\, dV_-.\]

Note, in particular, that if $\psi$ is supported near the bifurcation sphere $\mathcal{B}_+$ the  $\left\vert\left\vert\cdot\right\vert\right\vert_{\dot{H}^1\left(\mathcal{H}^-\right)}$ norm is finite if and only if $\psi$ lies in $\dot{H}^1_{\rm loc}$ along $\mathcal{H}^-\cup\mathcal{B}_+$. An analogous statement holds for the bifurcation sphere $\mathcal{B}_-$ and 
the $\left\vert\left\vert\cdot\right\vert\right\vert_{\dot{H}^1\left(\mathcal{CH}^+_B\right)}$ norm.
\end{definition}

Recall from Section~\ref{krusky} that $\mathcal{H}^-\cup\mathcal{B}_+\cup\mathcal{H}^+_B$ is a smooth null hypersurface. In particular, a solution $\psi$ in $\mathcal{M}$ is uniquely determined by its characteristic data along $\mathcal{H}^-\cup\mathcal{B}_+\cup\mathcal{H}^+_B$ and $\mathcal{C}^-$.

\begin{theorem}\label{frantheo}Let $\Psi : \mathcal{H}^- \to \mathbb{R}$ be smooth, compactly supported and spherically symmetric, and let $\psi$ be the unique solution to the wave equation
$(\ref{linearwaveequation})$ on $\mathcal{M}$ such that $\psi|_{\mathcal{H}^+_B\cup \mathcal{B}_+} = 0$, $\psi|_{\mathcal{H}^-} = \Psi$ and $\psi|_{\mathcal{C}^-} = 0$.

Then $\psi$ continuously extends to  $\mathcal{CH}^+_B \cup \mathcal{CH}^+_A \cup \mathcal{B}_-$ and satisfies
\begin{align*}
&\sup_{\mathcal{CH}^+_B \cup \mathcal{CH}^+_A \cup \mathcal{B}_-}\left|\psi\right|^2 + \int_{-\infty}^{\infty}\left|\frac{\partial\psi}{\partial v}|_{\mathcal{CH}^+_B}\right|^2\, dv  +  \int_{-\infty}^{\infty}\left|\frac{\partial \psi}{\partial u}|_{\mathcal{CH}^+_A}\right|^2\, du  \lesssim \left\vert\left\vert \Psi\right\vert\right\vert^2_{\dot{H}\left(\mathcal{H}^-\right) \cap L^{\infty}\left(\mathcal{H}^-\right)}.
\end{align*}
\end{theorem}
\begin{proof}This follows from a straightforward adaption of~\cite{Dafermos:2007jd}  and of~\cite{annefranzen} to  the Reissner--Nordstr\"{o}m--de Sitter setting.
\end{proof}

We now fix some notation for the ``seed data'' we will use to construct our desired solutions.

\begin{definition}\label{seed}Let $\Phi_{\mathcal{H}^-}(u) : \mathcal{H}^- \to \mathbb{R}$ be a non-zero smooth function compactly supported in $\{u \in (-1,0)\}$ and spherically symmetric, and let $\varphi$ be the corresponding spherically symmetric solution to the wave equation $(\ref{linearwaveequation})$ on $\mathcal{M}$ so that $\varphi|_{\mathcal{H}_-} = \Phi_{\mathcal{H}^-}$, $\varphi|_{\mathcal{H}^+_B\cup\mathcal{B}_+} = 0$ and $\varphi|_{\mathcal{C}^-} = 0$. 
\end{definition}

The next proposition is a statement of non-trivial reflection to $\mathcal{H}^+_A$.
\begin{proposition}\label{itwenttohplus} Let $\varphi$ denote the solution from Definition~\ref{seed}. Then we have 
\[\int_{-\infty}^{\infty}\left|\frac{\partial \varphi}{\partial v}|_{\mathcal{H}^+_A}\right|^2\, dv  > 0.\]
\end{proposition}
\begin{proof}This follows from Lemma~\ref{unique1} and a minor adaption of the proof of Theorem 10 from~\cite{Dafermos:2014jwa}.
\end{proof}

The next proposition is the statement that there is non-zero transmission to $\mathcal{CH}^+_B$.
\begin{proposition}\label{ittransmitsyay} Let $\varphi$ denote the solution from Definition~\ref{seed}. Then we have 
\[\int_{-\infty}^{\infty}\left|\frac{\partial \varphi}{\partial v}|_{\mathcal{CH}^+_B}\right|^2\, dv  > 0.\]
\end{proposition}
\begin{proof}This follows by combining Theorem~\ref{frantheo}, Proposition~\ref{itwenttohplus} and a minor adaption of the proof of Proposition 7.2 from \cite{DafShlap}.

\end{proof}

Finally, the next proposition is a manifestation of the time translation invariance of the spacetime.
\begin{proposition}\label{ittimetranslatesyay}Let $\Phi_{\mathcal{H}^-}$ and $\varphi$ be as in Definition~\ref{seed}.  For every $\tau > 0$, set 
\[\Phi^{(\tau)}_{\mathcal{H}^-} \doteq \Phi_{\mathcal{H}^-}\left(u-\tau\right),\] 
and let $\varphi^{(\tau)}$ be the corresponding solution to the wave equation
$(\ref{linearwaveequation})$ on $\mathcal{M}$ so that $\varphi^{(\tau)}|_{\mathcal{H}_-} = \Phi^{(\tau)}_{\mathcal{H}^-}$, $\varphi^{(\tau)}|_{\mathcal{H}^+_B\cup\mathcal{B}_+} = 0$ and $\varphi^{(\tau)}|_{\mathcal{C}^-} = 0$. 

Then, in the outgoing Eddington--Finklestein coordinates $(v,r,\theta,\phi)$, we have  $\varphi^{(\tau)}\left(v,r\right) = \varphi\left(v-\tau,r\right)$. In particular,
\[\varphi^{(\tau)}|_{\mathcal{CH}^+_B}\left(v\right) = \varphi|_{\mathcal{CH}^+_B}\left(v-\tau\right).\]

\end{proposition}
\begin{proof}This follows by combining Theorem~\ref{frantheo}, Proposition~\ref{itwenttohplus} and a minor adaption of the proof of Lemma 7.6 from \cite{DafShlap}.
\end{proof}

\subsection{From time translation invariance to scattering map amplification}\label{thekey}

The following proposition contains the essential content for the proof of Theorem~\ref{mainresult} and is closely related to the proof of Theorem 2 in~\cite{DafShlap}.
\begin{proposition}\label{theargheart} For a sufficiently small $\epsilon > 0$, there exists a one-parameter family of spherically symmetric and compactly supported $\Psi^{(\tau)}_{\mathcal{H}^-} : \mathcal{H}^- \to \mathbb{R}$  so that 
\begin{enumerate}
\item We have \[\sup_{\tau > 0}\left\vert\left\vert \Psi^{(\tau)}_{\mathcal{H}^-}\right\vert\right\vert_{H^{1+\epsilon}\left(\mathcal{H}^-\right)\cap L^{\infty}\left(\mathcal{H}^-\right)} \lesssim 1.\]
\item If we denote by $\psi^{(\tau)}$ the unique solution in $\mathcal{M}$ such that $\psi^{(\tau)}|_{\mathcal{H}^+_B\cup \mathcal{B}_+} = 0$, $\psi^{(\tau)}|_{\mathcal{H}^-} = \Psi^{(\tau)}_{\mathcal{H}^-}$ and $\psi^{(\tau)}|_{\mathcal{C}^-} = 0$, then
\[\sup_{\tau > 0}\left\vert\left\vert \psi^{(\tau)}\right\vert\right\vert_{\dot{H}^1\left(\mathcal{CH}^+_B\right)} = \infty.\]
\end{enumerate}

\end{proposition}
\begin{proof}
Recall the fundamental inequality $\kappa_+<\kappa_-$ satisfied by the surface gravities
(Lemma~\ref{fundsurfgrav}) and again chose thus $\hat{\kappa}$ satisfying
$\kappa_+<\hat\kappa<\kappa_-$. 
Let $\Phi$ and $\varphi$ be as in Definition~\ref{seed}. Define $\Psi^{(\tau)}_{\mathcal{H}^-}:\mathcal{H}^- \to \mathbb{R}$ by
\[\Psi^{(\tau)}_{\mathcal{H}^-}\left(u\right) \doteq e^{\frac{-\hat{\kappa}}{2}\tau}\Phi^{(\tau)}_{\mathcal{H}^-}\left(u\right) = e^{\frac{-\hat{\kappa}}{2}\tau}\Phi_{\mathcal{H}^-}\left(u-\tau\right),\]
and then let $\psi^{(\tau)}$ be the corresponding solution to the wave equation on $\mathcal{M}$ so that $\psi^{(\tau)}|_{\mathcal{H}_-} = \Psi^{(\tau)}_{\mathcal{H}^-}$,
 $\psi^{(\tau)}|_{\mathcal{H}^+_B\cup\mathcal{B}_+} = 0$ and $\psi^{(\tau)}|_{\mathcal{C}^-} = 0$. 

Clearly $\sup_{\tau > 0}\left\vert\left\vert \Psi^{(\tau)}_{\mathcal{H}^-}\right\vert\right\vert_{L^{\infty}\left(\mathcal{H}^-\right)} \lesssim 1$. 

Using Kruskal coordinates at $\mathcal{B}_+$ (see~\eqref{kruskalplus}), for large enough $\tau$ we may calculate
\begin{align}\label{h1psi}
\left\vert\left\vert \Psi^{(\tau)}_{\mathcal{H}^-}\right\vert\right\vert^2_{\dot{H}^1\left(\mathcal{H}^-\right)} &\lesssim  \int_{-1}^0\left|\frac{\partial \Psi^{(\tau)}_{\mathcal{H}^-}}{\partial U_+}\right|^2\, dU_+ 
\\ \nonumber &\lesssim \frac{1}{\kappa_+}\int_0^{\infty}\left|\frac{\partial}{\partial u}\Psi^{(\tau)}_{\mathcal{H}^-}\right|^2e^{\kappa_+ u}\, du
\\ \nonumber &= \frac{1}{\kappa_+}\int_{\tau-1}^{\tau}\left|\frac{\partial}{\partial u}\Phi_{\mathcal{H}^-}\left(u-\tau\right)\right|^2 e^{\kappa_+u - \hat{\kappa}\tau}\, du
\\ \nonumber &\sim e^{\left(\kappa_+-\hat{\kappa}\right)\tau}.
\end{align}

Similarly, we may estimate
\begin{align}\label{h2psi}
\left\vert\left\vert \Psi^{(\tau)}_{\mathcal{H}^-}\right\vert\right\vert^2_{\dot{H}^2\left(\mathcal{H}^-\right)} \lesssim e^{\left(3\kappa_+-\hat{\kappa}\right)\tau}. 
\end{align}
In particular, by interpolating, there must exist $\epsilon > 0$ so that
\begin{equation}\label{basich1epest1}
\sup_{\tau > 0}\left\vert\left\vert \Psi^{(\tau)}_{\mathcal{H}^-}\right\vert\right\vert^2_{\dot{H}^{1+\epsilon}\left(\mathcal{H}^-\right)} \lesssim 1.
\end{equation}

Next, our goal is to prove that
\begin{equation}\label{toblowup}
\sup_{\tau > 0}\left\vert\left\vert \psi^{(\tau)}\right\vert\right\vert^2_{\dot{H}^1\left(\mathcal{CH}^+_B\right)} = \infty.
\end{equation}
Proposition~\ref{ittransmitsyay} implies that there exists $v_0 \in (-\infty,\infty)$ and $c > 0$ so that
\begin{equation}\label{basiclowerbound}
\int_{v_0}^{v_0+1}\left|\frac{\partial\varphi}{\partial v}|_{\mathcal{CH}^+_B}\right|^2\, dv \geq c.
\end{equation}
Now, using~\eqref{basiclowerbound} and also Proposition~\ref{ittimetranslatesyay}, we have
\begin{align}\label{lowerboundforthewin}
\sup_{\tau > 0}\left\vert\left\vert \psi^{(\tau)}\right\vert\right\vert^2_{\dot{H}^1\left(\mathcal{CH}^+_B\right)} &\gtrsim \sup_{\tau > 0}\int_{-1}^0\left|\frac{\partial \psi^{(\tau)}}{\partial V_-}|_{\mathcal{CH}^+_B}\right|^2\, dV_- 
\\ \nonumber &= \frac{1}{\kappa_-}\sup_{\tau > 0}\int_0^{\infty}\left|\frac{\partial \phi}{\partial v}|_{\mathcal{CH}^+_B}\left(v-\tau\right)\right|^2 e^{\kappa_- v-\hat{\kappa}\tau}\, dv 
\\ \nonumber &=\frac{1}{\kappa_-}\sup_{\tau > 0}\int_{-\tau}^{\infty}\left|\frac{\partial \phi}{\partial v}|_{\mathcal{CH}^+_B}\left(v\right)\right|^2 e^{\kappa_- \left(v+\tau\right)-\hat{\kappa}\tau}\, dv 
\\ \nonumber &\geq \frac{1}{\kappa_-}\sup_{\tau > 0}\left[e^{\kappa_-\left(v_0+\tau\right)-\hat{\kappa}\tau}\int_{v_0}^{v_0+1}\left|\frac{\partial \phi}{\partial v}|_{\mathcal{CH}^+_B}\left(v\right)\right|^2 \, dv \right]
\\ \nonumber &\geq \frac{ce^{\kappa_-v_0}}{\kappa_-}\sup_{\tau > 0}e^{\left(\kappa_--\hat{\kappa}\right)\tau}
\\ \nonumber &=\infty.
\end{align}

\end{proof}

\subsection{The uniform boundedness principle and the proof of Theorem~\ref{mainresult}}\label{finishhim}
Finally we can use Proposition~\ref{theargheart} along with the uniform boundedess principle to prove Theorem~\ref{mainresult}.
\begin{proof}Let $\epsilon > 0$ be determined by Propsition~\ref{theargheart} and $\xi(x)$ be a smooth function which is identically $1$ for $x < 0$ and identically $0$ for $x > 1$. Also, let $H^{1+\epsilon}_{{\rm sph}}\left(\mathcal{H}^-\right) \cap L^{\infty}_{{\rm sph}}\left(\mathcal{H}^-\right)$ denote the Banach space of spherically symmetric functions in $H^{1+\epsilon}\left(\mathcal{H}^-\right) \cap L^{\infty}\left(\mathcal{H}^-\right)$.

Then, for every $N >0$, we define a map $\mathscr{T}_N: H^{1+\epsilon}_{{\rm sph}}\left(\mathcal{H}^-\right) \cap L^{\infty}_{{\rm sph}}\left(\mathcal{H}^-\right) \to \dot{H}^1\left(\mathcal{CH}^+_B\right) \cap L^{\infty}\left(\mathcal{CH}^+_B\right)$ by taking $\Psi_{\mathcal{H}^-} \in H^{1+\epsilon}_{{\rm sph}}\left(\mathcal{H}^-\right) \cap L^{\infty}_{{\rm sph}}\left(\mathcal{H}^-\right) $ to the corresponding unique solution $\psi$ to the wave equation~\eqref{linearwaveequation} which satisfies $\psi|_{\mathcal{H}^+_B\cup\mathcal{B}_+} = 0$, $\psi|_{\mathcal{H}^-} = \Psi_{\mathcal{H}^-}$ and $\psi|_{\mathcal{C}^-} = 0$, then  restricting $\psi$ to $\mathcal{CH}^+_B$ and then multiplying $\psi$ by $\xi\left(v-N\right)$, that is,
\[\mathscr{T}_N\left(\Psi_{\mathcal{H}^-}\right)\left(v\right) \doteq \xi\left(v-N\right)\psi|_{\mathcal{CH}^+_B}\left(v\right).\]
Note that the solutions $\psi$ constructed will be spherically symmetric weak solutions to the wave equation uniquely defined by the property that $\left(\psi,n_{\mathcal{S}}\psi\right) \subset
H^{1+\epsilon}_{{\rm loc}}(\mathcal{S})\times H^{\epsilon}_{{\rm loc}}(\mathcal{S})$ for any 
spacelike hypersurface $\mathcal{S} \subset \mathcal{M}$.

Theorem~\ref{frantheo} and a density argument implies that each $\mathscr{T}_N$ is a well-defined bounded map. Proposition~\ref{theargheart} implies that
\[\sup_N\sup_{\left\vert\left\vert \Psi_{\mathcal{H}^-}\right\vert\right\vert \leq 1}\left\vert\left\vert \mathscr{T}_N\Psi_{\mathcal{H}^-}\right\vert\right\vert_{\dot{H}^1_{{\rm sph}}\left(\mathcal{CH}^+_B\right)} = \infty.\]

The uniform boundedness principle (see Theorem III.9 of~\cite{ReedSimon}) and Theorem~\ref{frantheo} then imply that there exists $\Psi_{\mathcal{H}^-} \in H^{1+\epsilon}_{{\rm sph}}\left(\mathcal{H}^-\right) \cap L^{\infty}_{{\rm sph}}\left(\mathcal{H}^-\right)$ so that
\[\sup_N\left\vert\left\vert \mathscr{T}_N\Psi_{\mathcal{H}^-}\right\vert\right\vert_{\dot{H}^1_{{\rm sph}}\left(\mathcal{CH}^+_B\right)} = \infty.\]
This is, of course, equivalent to 
\[\left\vert\left\vert \psi\right\vert\right\vert_{\dot{H}^1\left(\mathcal{CH}^+_B\right)} = \infty.\]

Having established that the energy along $\mathcal{CH}^+_B$ blows up, the $H^1$ blow-up along any other null hypersurface $\mathcal{N}$ which intersects $\mathcal{CH}^+_A$ transversally follows from a standard propagation of singularities argument. (See, for example, Section 7.2.3 of~\cite{DafShlap}.)

Define $(\Psi,\Psi') \doteq \left(\psi|_{\Sigma},n_{\Sigma}\psi|_{\Sigma}\right)$. Using finite-in-time energy estimates and the spherical symmetry of $(\Psi,\Psi')$, it is straightforward to see that
\begin{equation}\label{basich1epest2}
\left\vert\left\vert (\Psi,\Psi') \right\vert\right\vert^2_{H^{1+\epsilon}\left(\Sigma\right) \times H^{\epsilon}\left(\Sigma\right)} \lesssim 1,
\end{equation}
and thus $\psi$ coincides in $D^+(\Sigma)$
with the solution produced by Proposition~\ref{local}.
We have thus constructed a solution $\psi$ in $D^+(\Sigma)$ 
arising from data $(\ref{basich1epest2})$
which satisfies the required blow up property at $\mathcal{CH}^+_A$,
completing the proof of Theorem~\ref{mainresult}.
\end{proof}

\section{Discussion}\label{discussionsec}

We end this paper by amplifying some of the comments already made in the introduction.

\paragraph{Additional angular regularity.}
We have already argued why for Christodoulou's formulation~\cite{Chr} 
of strong cosmic censorship,
it is clearly natural to also relax the regularity assumption on initial data.
Indeed, as we have mentioned, this has a precedent in the proof of weak cosmic censorship in spherical
symmetry~\cite{Christodoulou4}, where weak irregularities allowed in the
space of absolutely continuous functions were used
to maximally exploit the blue-shift instability connected to naked singularities.
On the other hand, one might object that for data which are only $H^{1+\epsilon}$,
one cannot show that general solutions to $(\ref{linearwaveequation})$ are continuous. 
To overcome this objection, it suffices to replace $H^{1+\epsilon}$ with a space where
additional regularity is imposed in the angular directions. In fact, 
in the Reissner--Nordstr\"om--de Sitter case, for a spherically symmetric hypersurface $\Sigma$, 
one can consider the space 
\begin{equation}
\label{anormsuchas}
\mathcal{D}_k(\Sigma)=\{(\Psi,\Psi'): \forall|\alpha|\le k,
\Omega^\alpha\Psi\in H^1(\Sigma), \Omega^\alpha\Psi'\in L^2(\Sigma)\}.
\end{equation}
Here $\alpha=(\alpha_1,\alpha_2,\alpha_3)$ is a multi-index and
$\Omega^\alpha$ denotes a string $\Omega_1^{\alpha_1}\Omega_2^{\alpha_2}\Omega_3^{\alpha_3}$
of angular momentum operators.
(Note that the data corresponding to an $H^1$ spherically symmetric solution $\psi_0$, and more generally, a solution $\psi_\ell$
supported on a fixed angular frequency $\ell$, is \emph{a fortiori} contained in 
$\mathcal{D}_k(\Sigma)$ for all $k$. Thus, our Theorem~\ref{mainresult} 
produces a solution
in $\mathcal{D}_k$ for all $k$.) 
For sufficiently high $k$,  this is precisely the norm considered
in~\cite{Dafermos:2007jd} for the wave equation on Schwarzschild--de Sitter.
This yields 
sufficient regularity and sufficiently fast decay in the region $\mathcal{M}_{\rm static}\cup
\mathcal{H}^+_A\cup\mathcal{C}^+$, in particular
along $\mathcal{H}^+_A$, so as to be able to  still apply for instance the results for 
$(\ref{linearwaveequation})$
on the black hole interior $\mathcal{M}_{\rm interior}$ region due 
to~\cite{annefranzen}.
(Note that one does not need in fact the extra $\epsilon$ and we have not included
it in the definition $(\ref{anormsuchas})$.)
Thus, solutions $\psi$ arising from data in $\mathcal{D}_k$ 
still share the same qualitative behaviour as smooth $C^\infty$ solutions. 
In particular, 
we can still moreover infer the continuous 
extendibility of $\psi$ beyond $\mathcal{CH}^+$.

\paragraph{Generalisation to $H^s$.}
Let us remark also that more generally, one can show for arbitrary $s\ge 1$,
and sufficiently small $\epsilon>0$, that given 
generic data $(\Psi,\Psi')\in H^{s(+\epsilon)}\times H^{s(+\epsilon)-1}$, the corresponding
solution $\psi$ of $(\ref{linearwaveequation})$
 is inextendible in
$H^s_{\rm loc}$ at the Cauchy horizon.

\paragraph{Extremal black holes with $\Lambda=0$.}
In the case $\Lambda=0$, it would be interesting to
revisit the  extremal
Reissner--Nordstr\"om spacetime from the point of view of the present paper.
Recall
that, for smooth localised initial data for $(\ref{linearwaveequation})$ on 
a suitable past Cauchy hypersurface $\Sigma$ 
crossing the event horizon $\mathcal{H}^+$
in extremal Reissner--Nordstrom,
it has been proven~\cite{Gajic:2015csa,Gajic:2015hyu,aag:2018extreme} that the resulting solutions
are indeed in $H^1_{\rm loc}$ at the Cauchy horizon.\footnote{This does not of course affect the validity of Christodoulou's formulation in $\Lambda=0$ since the extremal case is itself
non-generic.}
In view of the degeneration of the surface gravity,
it is not at all clear whether this failure can be circumvented by
 passing to a space of initial data of lower regularity.

 \paragraph{Decay rates and strong cosmic censorship for $C^\infty$ data?}
The above comments notwithstanding, in no way are we trying to argue that
one should cease investigation of the fine dynamics of  $C^\infty$ initial data. On the contrary!
Though we hope to have given
coherent  arguments for allowing for  a genericity assumption 
based on a norm such as $H^{1+\epsilon}$,  $(\ref{anormsuchas})$, or some other
related modification,
we do not view these arguments at present to be definitive. 
Unquestionably, 
the most
satisfactory resolution of the Christodoulou formulation of strong cosmic censorship would  be one independent
of the precise regularity assumptions made at the level of initial data, and the cleanest way to have this would be for the
result to have been true in the topology of the $C^\infty$ class.  If Christodoulou's
formulation of strong cosmic censorship indeed fails for the Einstein--Maxwell
system $(\ref{eq:einsteinmaxwell})$ with  $\Lambda>0$ in the smooth topology,
then a nagging dissatisfaction with the whole situation is still inevitable, despite
our proposed circumvention.
In any case, irrespectively of its final significance for strong cosmic censorship, the 
problem of understanding the fine asymptotics of solutions arising from $C^\infty$ initial 
data---and the implications of this for the sharp generic inextendibility statement at
the Cauchy horizon---is
certainly an extremely worthwhile open problem that very much remains to be properly  understood.
\bibliographystyle{DHRalpha}
\bibliography{finalref}

\end{document}